\documentclass[bj,authoryear,noshowframe]{imsart}

\graphicspath{{figures/}}
\usepackage{subcaption}

\usepackage{natbib}

\startlocaldefs
\theoremstyle{plain}

\newtheorem{theorem}{Theorem}[section]
\newtheorem{lemma}[theorem]{Lemma}
\newtheorem{proposition}[theorem]{Proposition}
\newtheorem{corollary}[theorem]{Corollary}
\theoremstyle{remark}

\newtheorem{assumption}[theorem]{Assumption}
\newtheorem{remark}{Remark}[section]

\newcommand{\added}[1]{#1}
\newcommand{\replaced}[2]{#1}
\newcommand{\deleted}[1]{}

\newcommand{\suppA}{A}
\newcommand{\suppB}{B}
\newcommand{\suppC}{C}
\newcommand{\suppEqBPSODE}{19}
\newcommand{\dif}{\mathrm{d}}
\newcommand{\arctanh}{\operatorname{arctanh}}
\newcommand{\former}{\mathrm{K}}
\newcommand{\latter}{\mathrm{L}}
\newcommand{\vertiii}[1]{{\left\vert\kern-0.25ex\left\vert\kern-0.25ex\left\vert #1 
    \right\vert\kern-0.25ex\right\vert\kern-0.25ex\right\vert}}

\endlocaldefs

\begin{document}

\begin{frontmatter}
\title{Scaling of piecewise deterministic Monte Carlo for anisotropic targets}
\runtitle{Scaling of PDMC for anisotropic targets}

\begin{aug}
\author[A]{\fnms{Joris}~\snm{Bierkens}\ead[label=e1]{joris.bierkens@tudelft.nl}\orcid{0000-0003-0185-5804}}
\author[B]{\fnms{Kengo}~\snm{Kamatani}\ead[label=e2]{kamatani@ism.ac.jp}\orcid{0000-0002-1537-3446}}
\author[C]{\fnms{Gareth O. }~\snm{Roberts}\ead[label=e3]{Gareth.O.Roberts@warwick.ac.uk}}
\address[A]{Delft Institute of Applied Mathematics\printead[presep={,\ }]{e1}}

\address[B]{Institute of Statistical Mathematics\printead[presep={,\ }]{e2}}

\address[C]{Department of Statistics,
University of Warwick\printead[presep={,\ }]{e3}}
\end{aug}

\begin{abstract}
Piecewise deterministic Markov processes (PDMPs) are a type of continuous-time Markov process that combine deterministic flows with  jumps. Recently, PDMPs have garnered attention within the Monte Carlo community as a potential alternative to traditional Markov chain Monte Carlo (MCMC) methods. The Zig-Zag sampler and the Bouncy Particle Sampler are commonly used examples of the PDMP methodology which have also yielded impressive theoretical properties, but little is known about their robustness to extreme dependence or \replaced{anisotropy}{isotropy} of the target density.
It turns out that PDMPs may suffer from poor mixing due to anisotropy and this paper investigates this effect in detail in the stylised but important Gaussian case. 
To this end, we employ a multi-scale analysis framework in this paper. Our results show that when the Gaussian target distribution has two scales, of order $1$ and $\epsilon$, the computational cost of the Bouncy Particle Sampler is of order $\epsilon^{-1}$, and the computational cost of the Zig-Zag sampler is \replaced{$\epsilon^{-2}$}{either $\epsilon^{-1}$ or $\epsilon^{-2}$, depending on the target distribution}. In comparison, the cost of the traditional MCMC methods such as RWM 
is of order $\epsilon^{-2}$, at least when the dimensionality of the small component is more than $1$. Therefore, there is a robustness advantage to using PDMPs in this context.
\end{abstract}



\end{frontmatter}

\section{Introduction}

\label{sect:intro}

Piecewise deterministic Markov processes (PDMPs) are continuous-time non-reversible Markov processes driven by deterministic flows and jumps. Recently, PDMPs have attracted \added{the} attention of the Monte Carlo community as a possible alternative to traditional Markov chain Monte Carlo methods. The Zig-Zag sampler \cite{BierkensFearnheadRoberts2016} and the Bouncy Particle Sampler \cite{PhysRevE.85.026703, BouchardCote2017} are commonly used examples of the PDMP methodology \added{to approximate target distributions defined in $\mathbb{R}^d$}. \added{Both Markov samplers operate in subsets of $\mathbb{R}^d\times\mathbb{R}^d$, with states represented as $(x,v)$. When these processes achieve ergodicity, the empirical averages of the Markov processes in the $x$-space converge to the corresponding expectations under the original target distribution, by the law of large numbers. }

These samplers have yielded impressive theoretical properties, for example in high-dimensional contexts \cite{deligiannidisrandomized,andrieu2021hypocoercivity,bierkens2018high}, \added{in terms of asymptotic variance reduction \cite{BierkensDuncan2016}, and} in the context of MCMC with large data sets \cite{BierkensFearnheadRoberts2016}.



It is well-known that all MCMCs suffer from deterioration of mixing properties in the context of anisotropy of the target distribution. For example the mixing of random walk Metropolis, Metropolis-Adjusted Langevin and Hamiltonian schemes can all become arbitrarily bad in any given problem for increasingly inappropriately scaled proposal distributions (see for example 
the studies in \cite{RGG,roberts1998optimal,roberts2001optimal,beskos2013optimal,beskos2018, Graham2022}). Meanwhile the Gibbs sampler is known to perform poorly in the context of highly correlated target distributions, see for example \cite{robsah97,zanella2021multilevel}.  In principle these problems can be overcome using preconditioning (a priori reparameterisations of the state space), often carried out using adaptive MCMC  \citep{andrieu2008tutorial,roberts2009examples}. However, effective preconditioning can be very difficult in higher-dimensional problems.  Therefore, it is important to understand the extent of this deterioration of mixing in the face of anisotropy or dependence.

PDMPs can also suffer from poor mixing due to anisotropy and this paper will study this effect in detail in the stylised but important Gaussian case. In principle, continuous time processes are well-suited to this problem because they do not require a rejection scheme, unlike traditional MCMC methods which suffer from a scaling issue. However, this does not guarantee that the processes will be effective for this problem. Even if the process mixes well, it may require a large number of jumps, resulting in high computational cost. To understand this quantitatively, it is important to determine the influence that anisotropies in the target distribution have on the speed convergence to equilibrium.

However, for continuous-time processes a further consideration relevant to the computational efficiency of an algorithm involves some measure of the computational \replaced{effort}{expense} expended per unit stochastic process time. It is difficult to be precise about how to measure this in general. Here we will use a commonly adopted surrogate: the number of jumps per unit time. 
Effectively this assumes that we have access to an efficient implementation scheme for our PDMP via an appropriate Poisson thinning scheme, which is not always available in practice \added{(see Section \ref{sec:2dim-discussion} for the detail)}.

This quantification of the computational cost can be achieved by the scaling analysis which first appeared in \cite{RGG} in the literature under mathematically formal argument. 
In this paper, we follow  \cite{beskos2018} that studied the random walk Metropolis algorithm for anisotropic target distributions. To avoid technical difficulties, they considered a simple target distribution with both small and large-scale components, where these two components are conditionally independent. 
The scale of the noise of the algorithm must be appropriately chosen in advance, otherwise, the algorithm will be frozen. On the other hand, the piecewise deterministic Markov processes naturally fit the scale without artificial parameter tuning. However, from a theoretical point of view, this makes the analysis complicated since the process changes its dynamics depending on the target distribution. Therefore, we study a simple Gaussian example and identify the factors that change the dynamics for a better understanding of the processes for anisotropic target distributions. 

The main contributions of this paper are  thus  the following:
\begin{enumerate}
\item
We give a theoretical and rigorous study of the robustness of  Zig-Zag and BPS algorithms in the case of increasingly extreme anisotropy (where one component is scaled by a factor $\epsilon$ which converges to $0$ in the limit). \replaced{Specifically, we derive weak convergence results that}{The results give weak convergence results to} characterise the limiting dynamics of the slow component in the extreme anisotropy limit.
\item
For the Zig-Zag, we see that, unless the slow component is either (a) perfectly aligned with one of the possible velocity directions, or (b) independent of the fast component, the limiting behaviour is a reversible Ornstein--Uhlenbeck process
\replaced{that mixes in time $O(\epsilon^{-1})$. Considering that the number of switches per unit time scales as $O(\epsilon^{-1})$ (by Theorem \ref{theo:zz-costs}), the total complexity grows as $O(\epsilon^{-2})$. }{
(\ref{eq:zz-ou})
as presented in Theorem \ref{theo:zz-main}, and which
mixes in time $O(\epsilon ^{-1})$ but for which $O(\epsilon ^{-1})$ switches are required by Theorem \ref{theo:zz-costs} leading to a complexity which is $O(\epsilon ^{-2})$.}
\item
For the BPS \added{in Theorem \ref{theo:bps-main} we show that its} limiting behaviour is a deterministic dynamical system for which we can explicitly write down its generating ODE. 
\deleted{(19) in the supplementary material as presented in Theorem \ref{theo:bps-main}.}
\replaced{Mixing occurs in $O(1)$ time, which when combined with a cost of $O(\epsilon^{-1})$ switches per unit time, yields an overall complexity of $O(\epsilon^{-1})$. }{
  Mixing is $O(1)$ and the number of  switches required for this is  $O(\epsilon ^{-1})$ switches by Theorem \ref{theo:bps-costs} leading to a complexity which is $O(\epsilon ^{-1})$. }
\end{enumerate}

The paper is organised as follows. In Section \ref{sect:twod} we explore the problem in a 2-dimensional setting where we can easily illustrate and motivate our results. Here we describe the limiting processes using the homogenisation and averaging techniques of \cite{Pavliotis2008}. Section  \ref{sect:main} then describes our main results in detail and signposts the proofs. \replaced{Section \ref{sec:leading}}{Sections  \suppA  and  \suppB}  then gives the detailed proofs for the Zig-Zag and BPS respectively, supported by technical material from \replaced{the supplementary material \citep{BierkensKamatanRoberts2024}}{Section \suppC}. \added{More specifically, in Section \ref{sec:leading}, we illustrate the properties of the leading-order generator. Additionally, in the supplementary material, we demonstrate weak convergence for the Zig-Zag sampler through homogenisation, and for the Bouncy Particle Sampler through averaging methods}. A final discussion is given in Section  \ref{sect:disc}.

\deleted{
Recent work  by Andrieu et al. (2021) provides $L^2$-exponential convergence for general target distributions, identifying the exponential rate even in high-dimensional scenarios. Their results can also be applied to anisotropic target distributions. However, it is worth noting that it is currently unclear how to achieve the same rate for anisotropic target distributions as described above.}

\section{Exhibit of main results in the two-dimensional case}
\label{sect:twod}



To better understand our result, we first consider the two-dimensional case, which is  easier to interpret than the general result. Consider the situation of a \added{centered} two-dimensional Gaussian target distribution with one large and one small component. Indeed, let $\Sigma^\epsilon = U^{\top} (\Lambda^\epsilon)^2 U$, where
\begin{equation}
\label{eq:2dim-unitary-matrix}
    U = \begin{pmatrix} \cos \theta & - \sin \theta \\ \sin \theta  & \cos \theta \end{pmatrix} \quad \text{and} \quad \Lambda^\epsilon = \begin{pmatrix} 1 & 0 \\ 0 & \epsilon \end{pmatrix}.
\end{equation}
We will show that the Markov processes of the Zig-Zag sampler and the Bouncy Particle Sampler converges to limit processes as $\epsilon\rightarrow 0$ if we choose a correct scaling in space and time. 

To illustrate the impact of $\epsilon$, we will use a reparametrised variable for our scaling analysis: 
\begin{equation}\label{eq:reparametrise}
    y=(\Lambda^\epsilon)^{-1} Ux. 
\end{equation}
By this variable scaling, the invariant distribution of the Markov semigroup corresponding to $\xi=(y,v)$ is $\mu=\mathcal{N}_d(0,I_d)\otimes\mathcal{U}(\{-1,+1\}^d)$ for the Zig-Zag sampler, and $\mu=\mathcal{N}_d(0,I_d)\otimes\mathcal{N}_d(0, I_d)$ for the Bouncy Particle Sampler. 

\subsection{The Zig-Zag sampler in the two dimensional case}
\label{subsec:zz-digest}

The Markov generator of the two-dimensional Zig-Zag sampler with reparametrisation is
\begin{equation}
\label{eq:generator-zigzag-2d}
\begin{split}
    \mathcal{L}^\epsilon f(y,v)&=
    (v_1 \cos\theta-v_2 \sin\theta) \partial_{y_1} 
    f(y,v)+\epsilon^{-1}(v_1\sin\theta +v_2\cos\theta ) \partial_{y_2} 
    f(y,v)\\
    &\quad
    + \left(v_1(y_1\cos\theta+\epsilon^{-1}y_2\sin\theta)\right)_+~(\mathcal{F}_1-\operatorname{id})f(y,v)\\
   &\quad+ \left(v_2(-y_1\sin\theta+\epsilon^{-1}y_2\cos\theta)\right)_+~(\mathcal{F}_2-\operatorname{id})f(y,v), 
\end{split}
\end{equation}
where $\mathcal{F}_i f(y,v)=f(y, F_i(v))$ and $F_i$ represents the operation that flips the sign of the $i$-th component. 
Figure \ref{fig:zz2d} displays the trajectories of the Zig-Zag sampler in a two-dimensional scenario, with $\theta = 0, \pi/6$ and $\pi/4$ and $\epsilon=0.01$. 

\begin{figure}[h]
\centering
\includegraphics[width=0.28\linewidth]{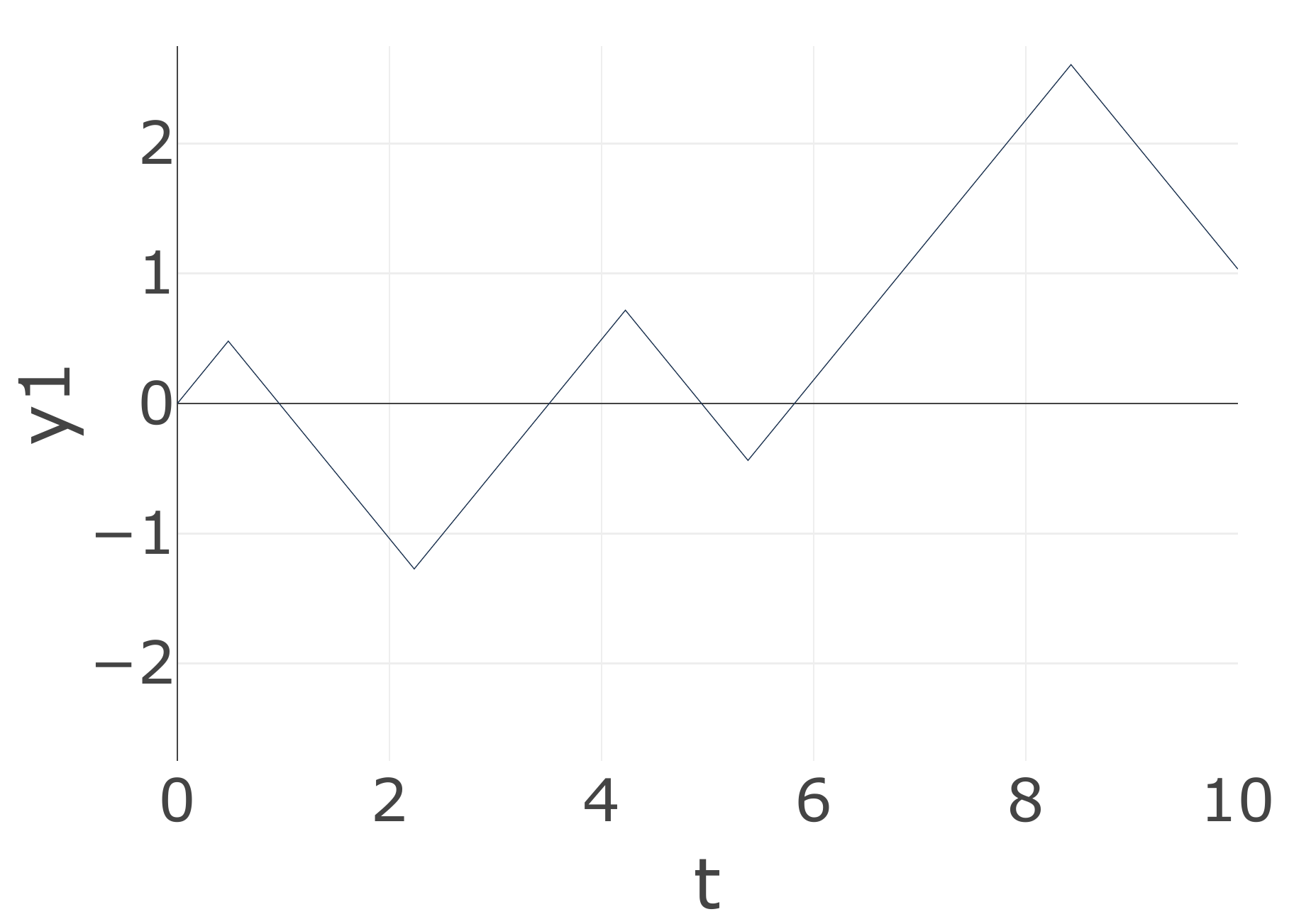}
\includegraphics[width=0.28\linewidth]{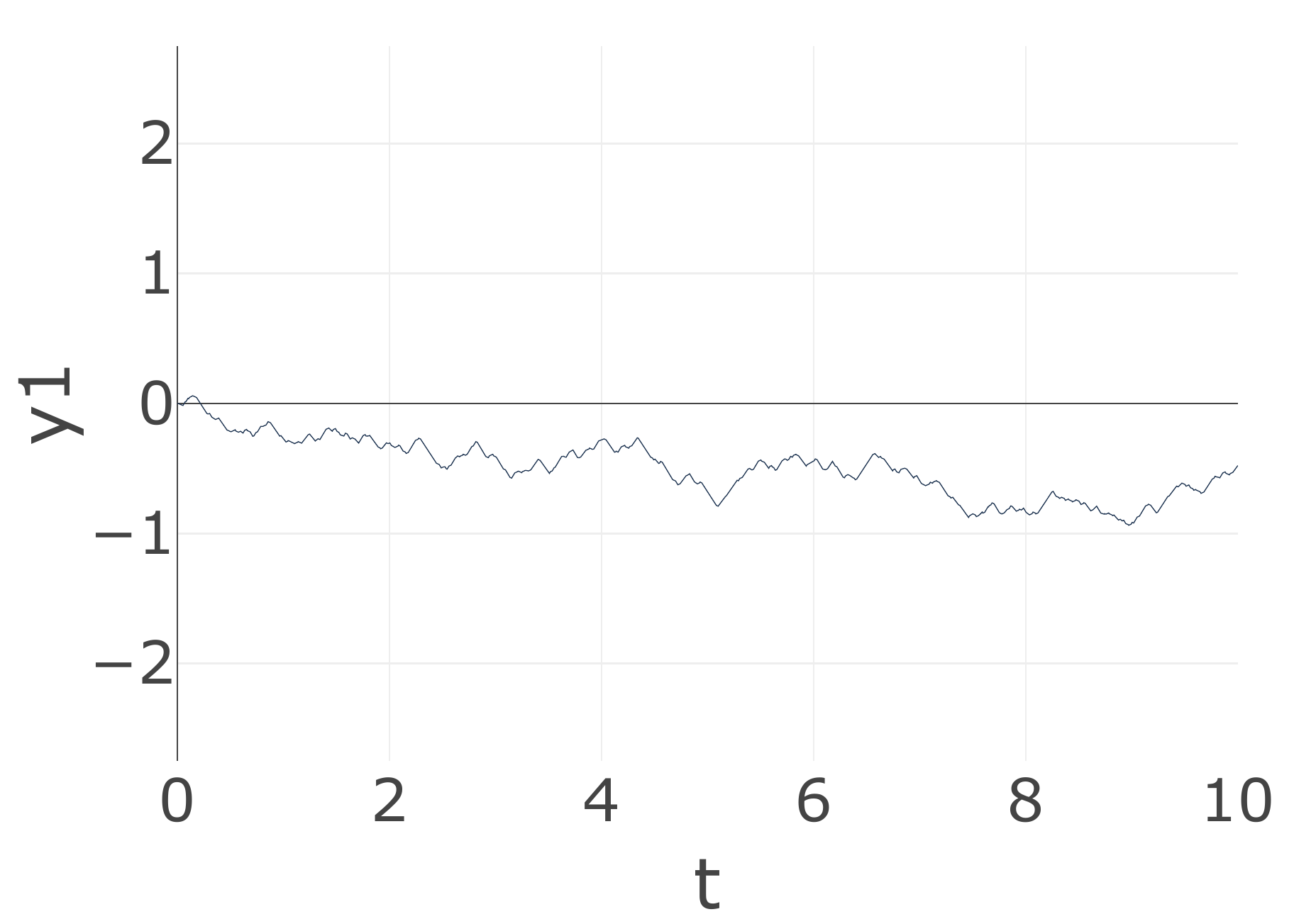}
\includegraphics[width=0.28\linewidth]{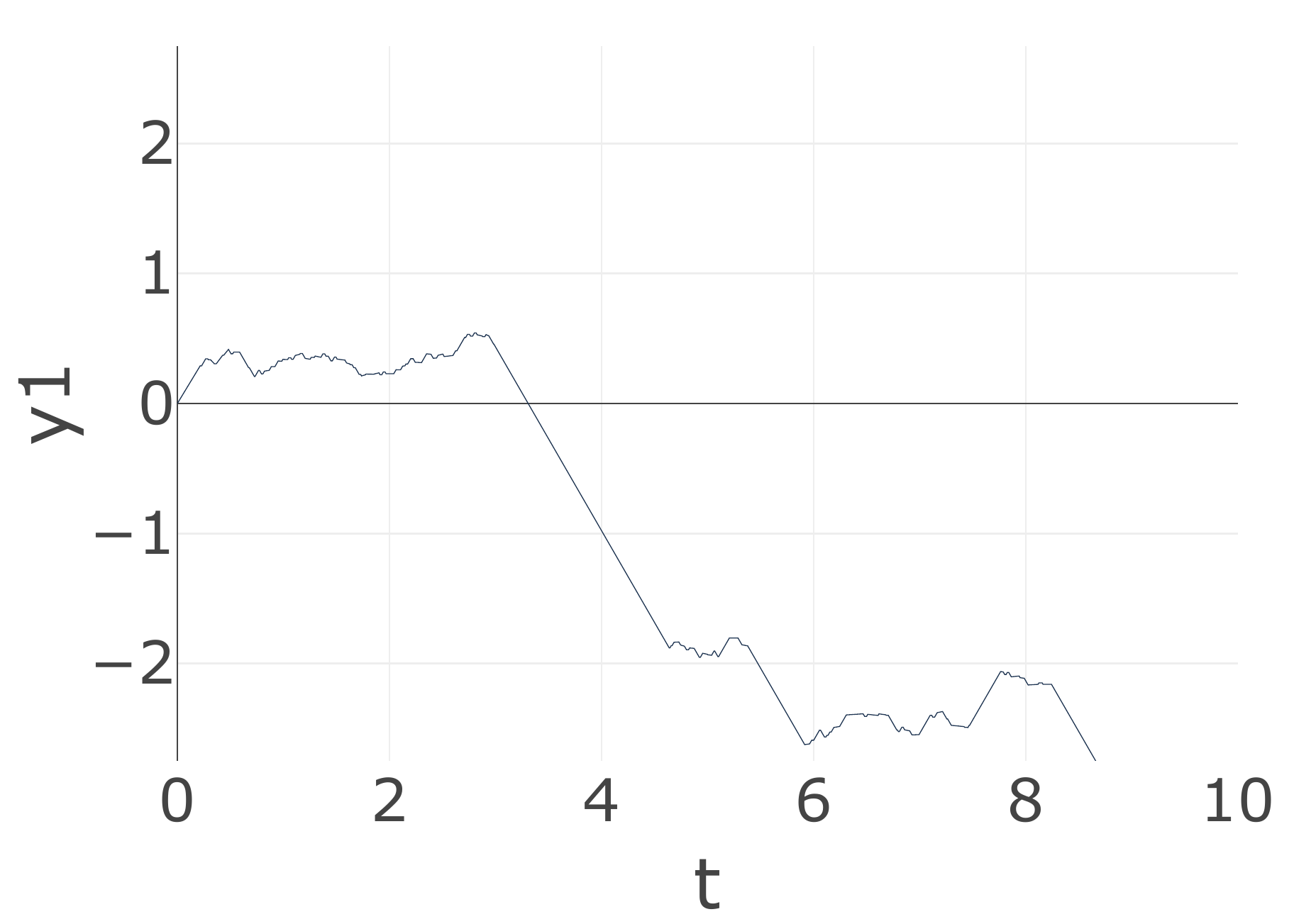}\\
\includegraphics[width=0.28\linewidth]{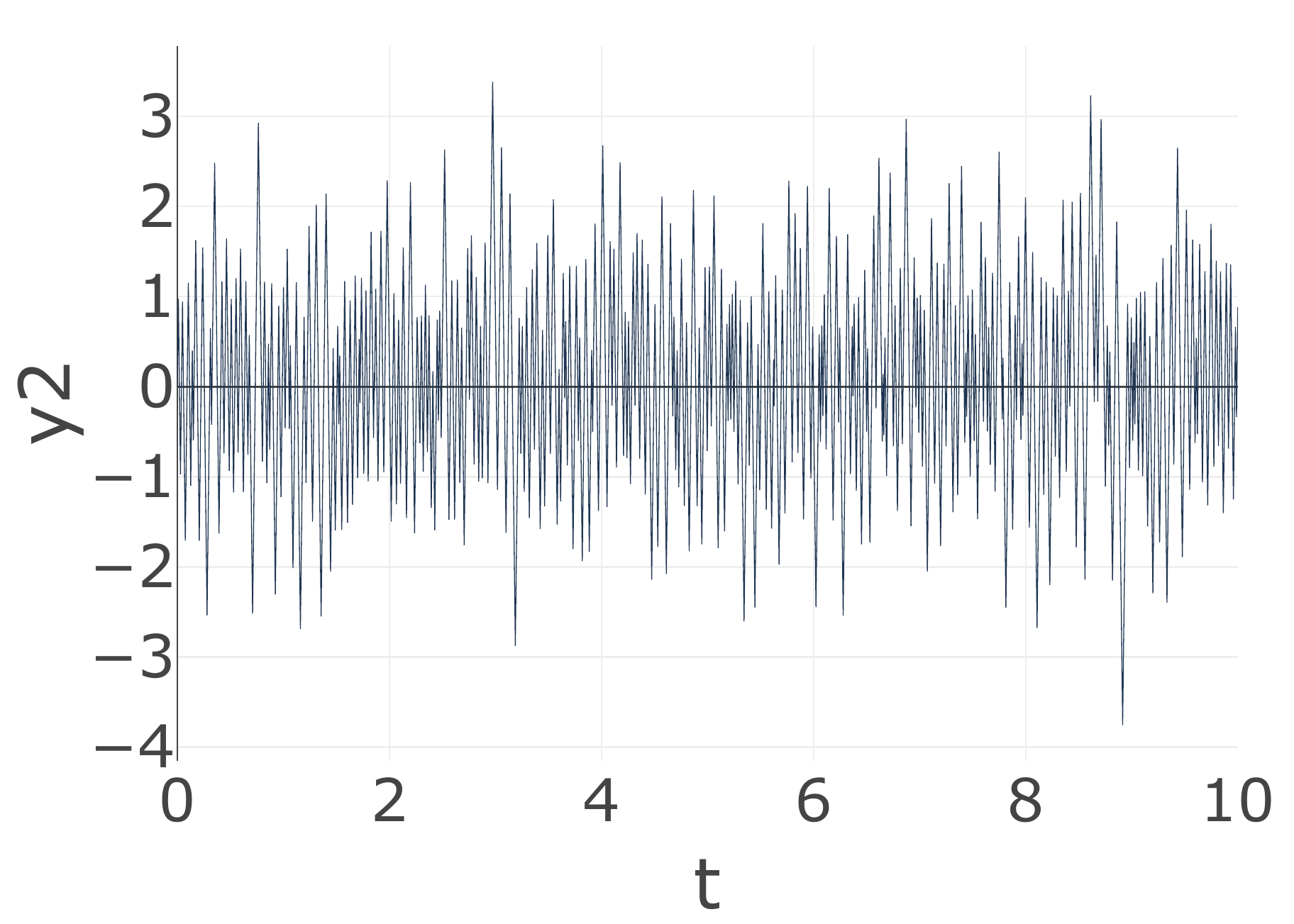}
\includegraphics[width=0.28\linewidth]{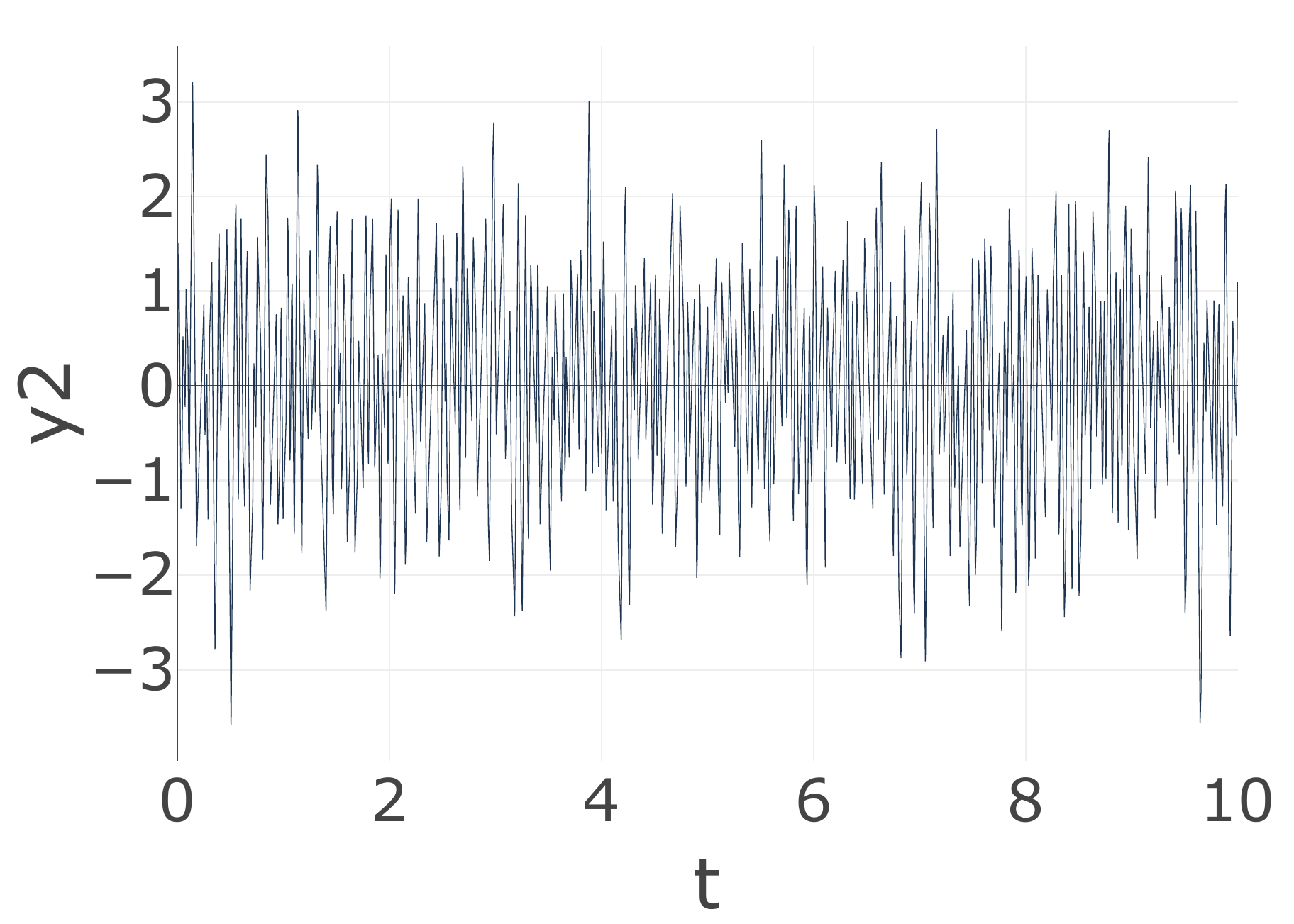}
\includegraphics[width=0.28\linewidth]{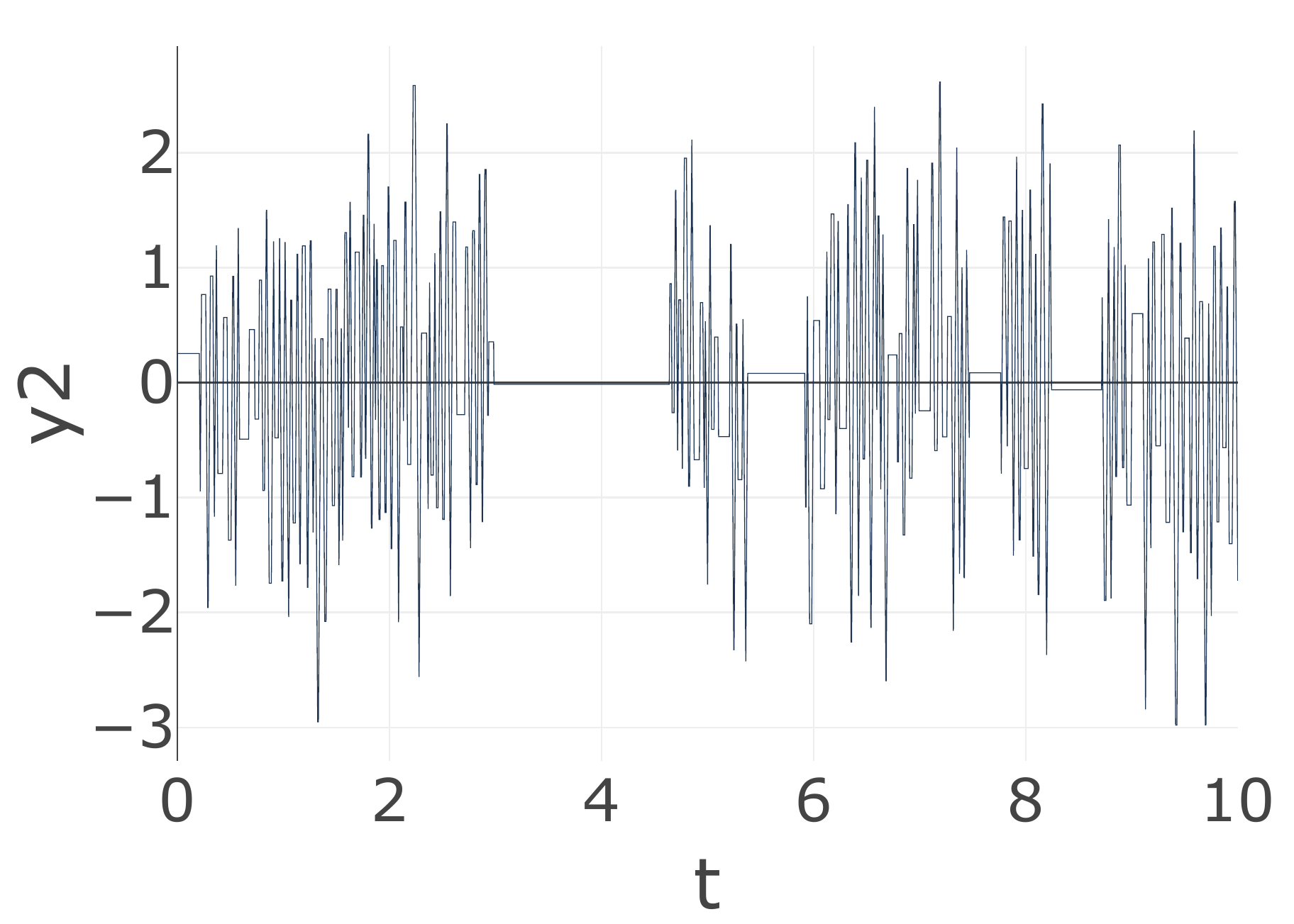}\\
\includegraphics[width=0.28\linewidth]{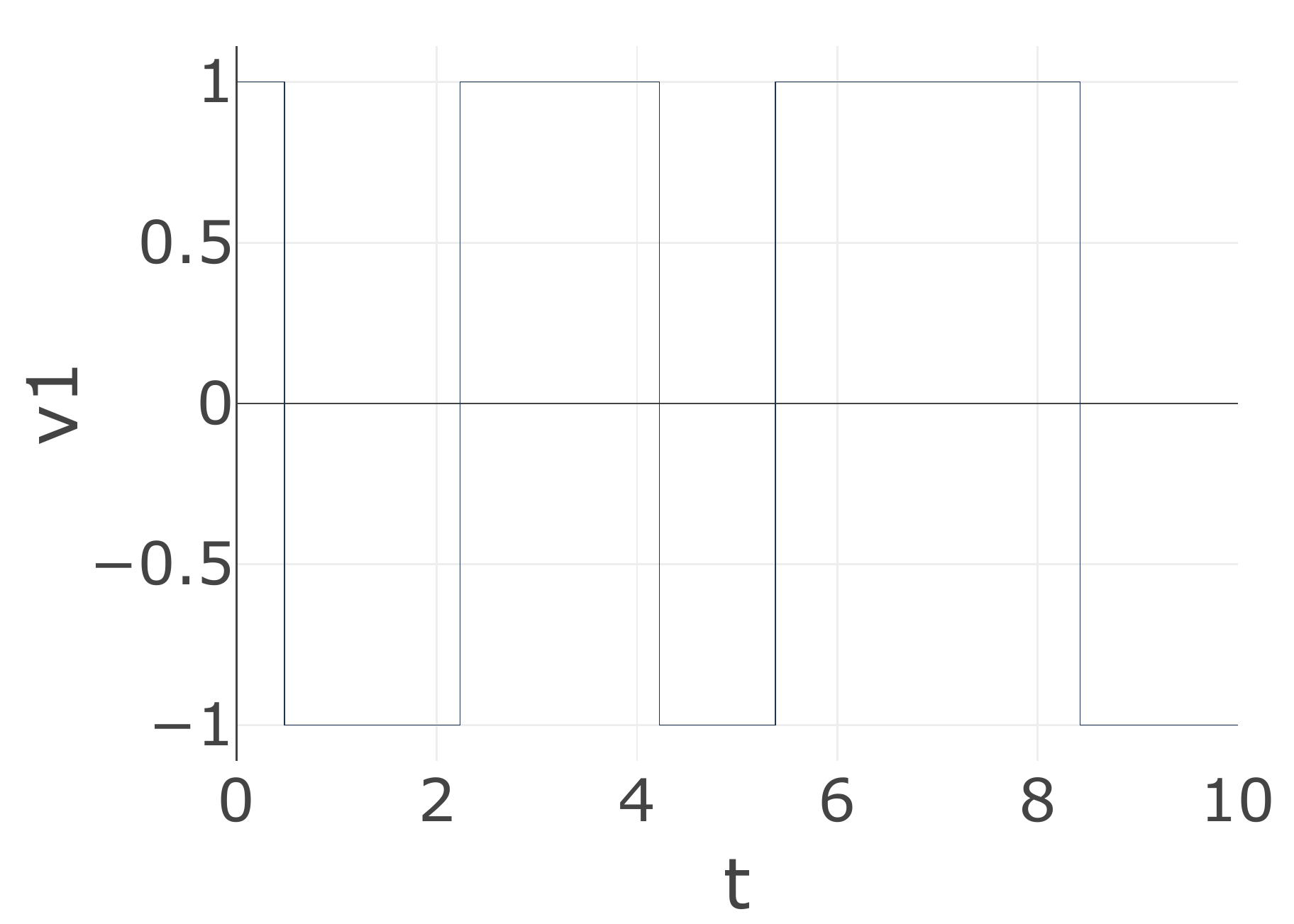}
\includegraphics[width=0.28\linewidth]{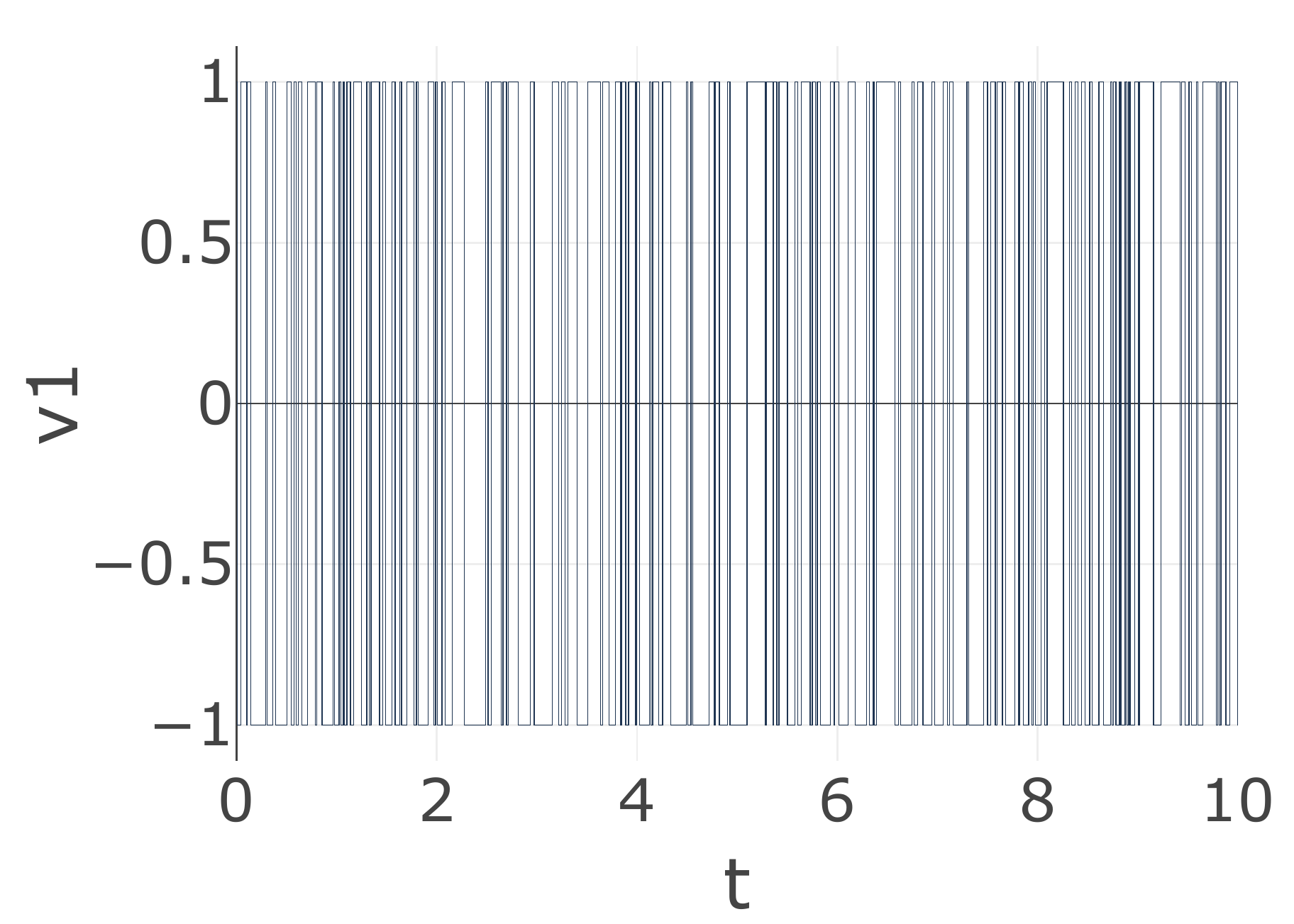}
\includegraphics[width=0.28\linewidth]{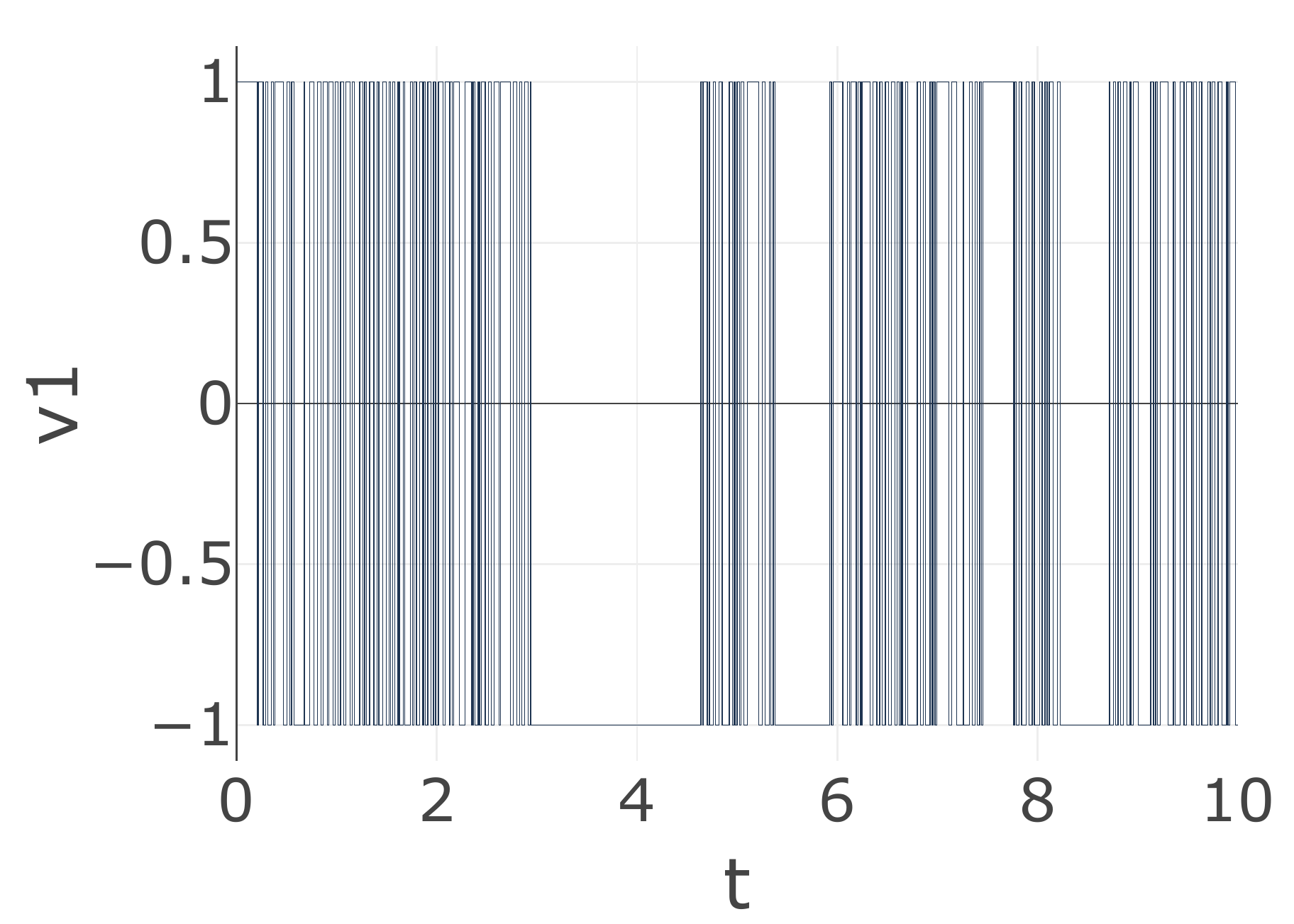}\\
\includegraphics[width=0.28\linewidth]{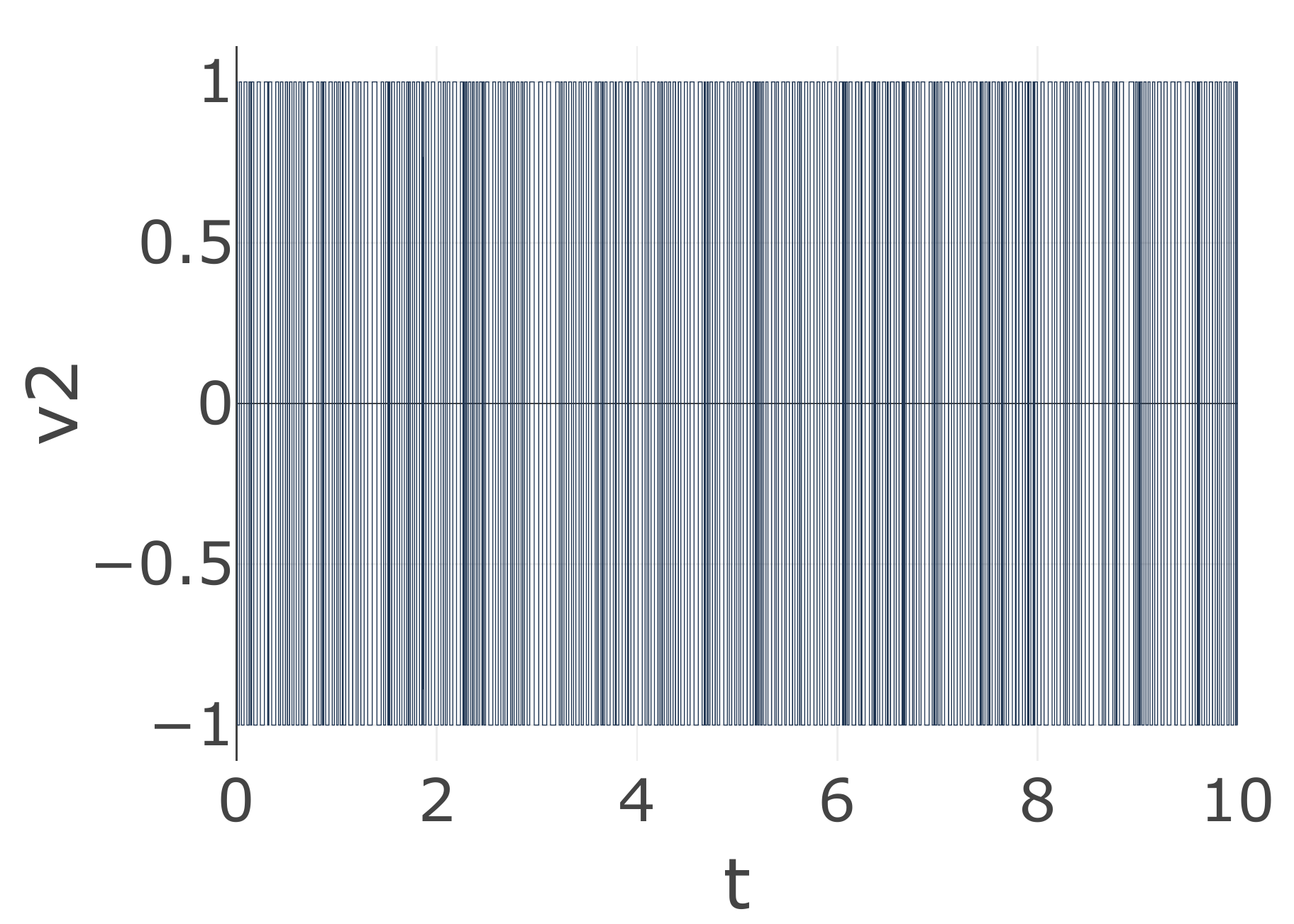}
\includegraphics[width=0.28\linewidth]{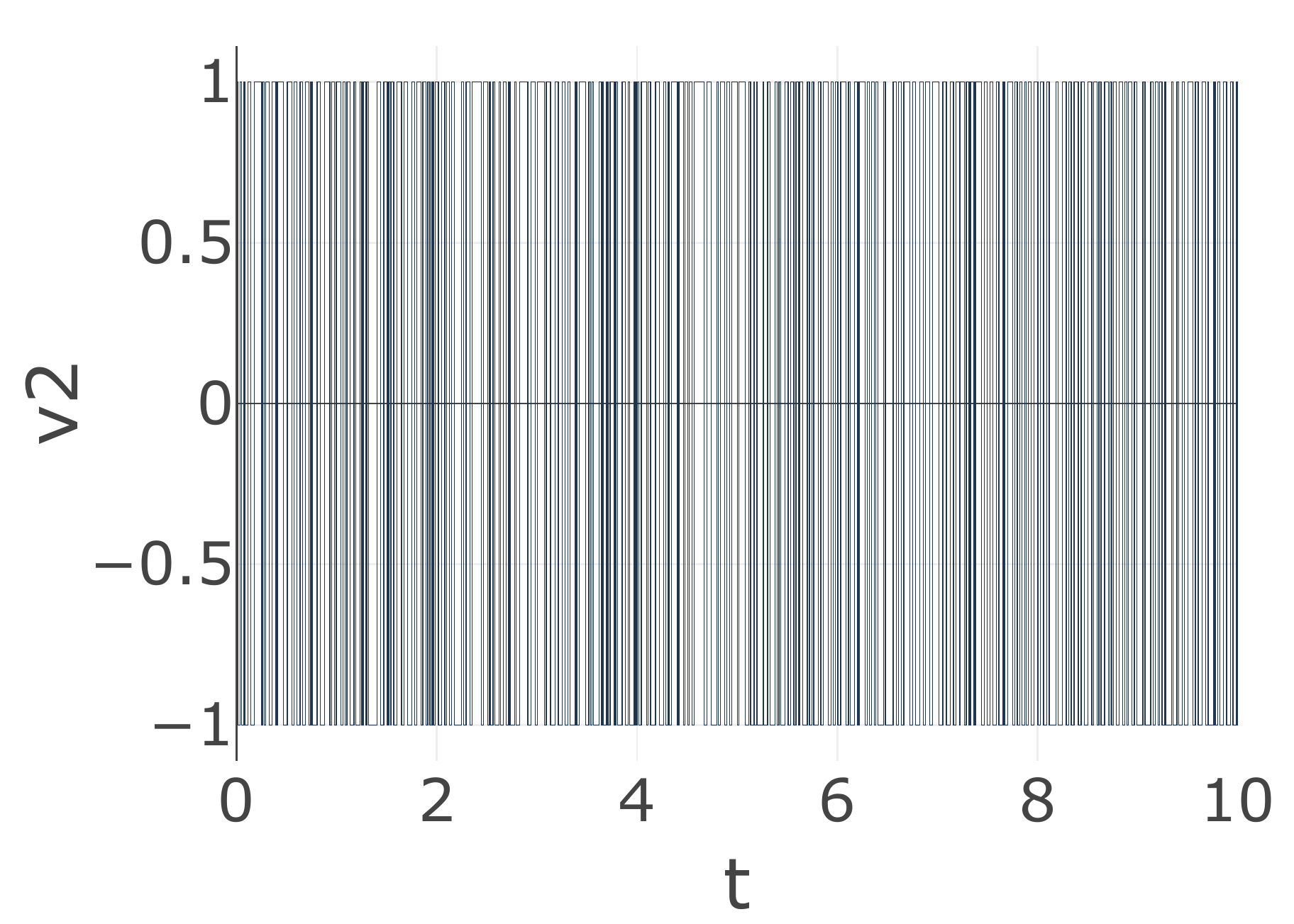}
\includegraphics[width=0.28\linewidth]{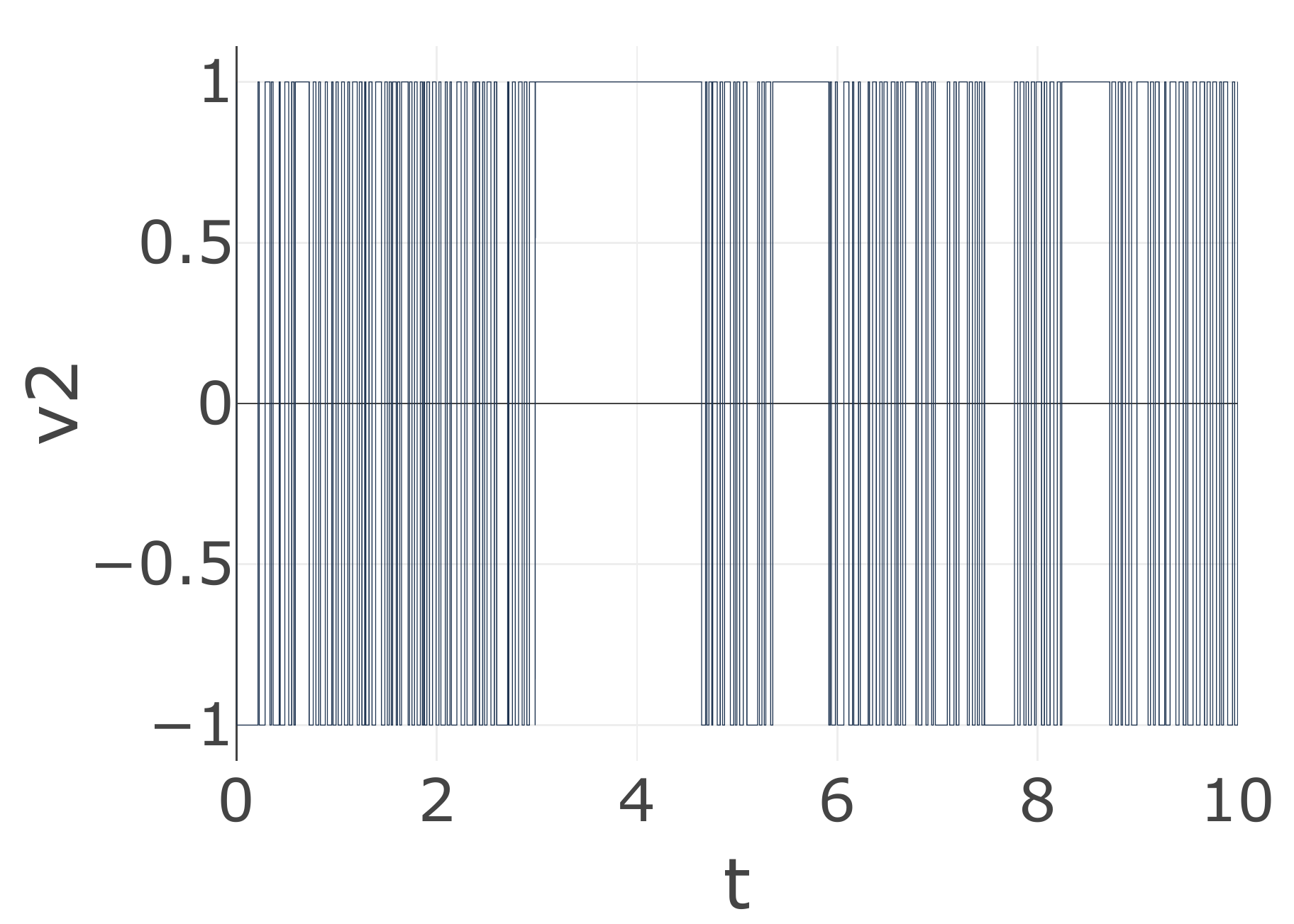}
\caption{Paths of $y_1$ for $\theta=0$ (left), 
$\theta=\pi/6$ (centre), and
$\theta=\pi/4$ (right), and $y_1$ ($1$st row), $y_2$ ($2$nd row) and $v_1,v_2$ ($3$rd and $4$th rows) of the Zig-Zag sampler where  $\epsilon=0.01$. 
}
\label{fig:zz2d}
\end{figure}

Our first observation from Figure \ref{fig:zz2d} is that the process of $(y_2,v_1,v_2)$ exhibits good mixing, whereas the process of $y_1$ does not. This result is expected, as the terms in the expression for $\mathcal{L}^\epsilon$ that correspond to the dynamics of $(y_2,v_1,v_2)$ are proportional to $\epsilon^{-1}$, whereas the dynamics of $y_1$ are of order $O(1)$. When $\epsilon$ is sufficiently small, the dynamics of $(y_2,v_1,v_2)$ are described by the operator $\mathcal{L}_0$, which is derived from $\mathcal L^{\epsilon}$ by retaining only terms proportional to $\epsilon^{-1}$:
\begin{align*}
\mathcal{L}_0f(y,v)
&= 
    (v_1\sin\theta +v_2\cos\theta )~ \partial_{y_2} 
    f(y,v)\\
&\quad  +\left(v_1y_2\sin\theta\right)_+~(\mathcal{F}_1-\operatorname{id})f(y,v)
    + \left(v_2y_2\cos\theta\right)_+~(\mathcal{F}_2-\operatorname{id})f(y,v). 
\end{align*}


Our second observation is that the Zig-Zag process corresponding to $y_1$ generally demonstrates diffusive behaviour, except in some specific situations, which correspond to the figures on the left and right in Figure \ref{fig:zz2d}. The behaviour in these exceptional scenarios will be discussed in Section \ref{sec:diagonally_aligned}. For the time being, we make the following assumption.

\begin{assumption}
\label{ass:zz2d}
$\theta\in [0, 2\pi)$ but $\theta\neq n\pi/4\ (n=0,\ldots, 7)$. 
\end{assumption}

To clarify the asymptotic behavior, the solution of the Poisson equation plays a crucial role, as it enables the separation of the fast dynamics associated with the $y_2$ process and the slow dynamics associated with the $y_1$ process.  
Let $\chi(y_2, v)$ denote the solution to the Poisson equation
\begin{equation}
\label{eq:2dpoisson}
    \mathcal{L}_0\chi(y_2,v)=-(v_1\cos\theta-v_2\sin\theta).
\end{equation}
Let $\Omega=-2\mu(\chi~\mathcal{L}_0\chi)$ which can be shown to admit the expression 
\begin{equation}
\label{eq:omega}
    \Omega= 
    \begin{cases}
    \dfrac{8}{\sqrt{\pi}}\dfrac{\arctanh\sqrt{|\tan\theta|}}{(1+|\sin 2\theta|)\sqrt{|\sin 2\theta|}}+\dfrac{\sqrt{2\pi}}{|\sin\theta|(1+|\sin 2\theta|)}& \mathrm{if}\  |\sin\theta |<|\cos\theta|\\
    \vspace{1px}\\
    \dfrac{8}{\sqrt{\pi}}\dfrac{\arctanh\sqrt{1/|\tan\theta|}}{(1+|\sin 2\theta|)\sqrt{|\sin 2\theta|}}+\dfrac{\sqrt{2\pi}}{|\cos\theta|(1+|\sin 2\theta|)}&\mathrm{if}\  |\sin\theta |>|\cos\theta|
    \end{cases}
\end{equation}
for $\theta\neq n\pi/4\ (n=0,\ldots, 7)$.  See Figure \ref{fig:zz2defficiency}. 
The expression of $\chi$, along with the derivation of $\Omega$, can be found in Section \suppC. 

Let $y_1^\epsilon(t)$ be the $y_1$-coordinate of the Zig-Zag sampler process. 
The following is the formal statement of our second observation. 
Let $\mathbb{D}[0,T]$ be the space of c\`adl\`ag processes on $[0,T]$ equipped with the Skorokhod topology. 

\begin{theorem}
Under Assumption \ref{ass:zz2d}, and when the process is stationary, the $\epsilon^{-1}$-time scaled process
$y_1^\epsilon(\epsilon^{-1}t)$ converges weakly in $\mathbb{D}[0,T]$ for any $T>0$ to the  Ornstein--Uhlenbeck process
$$
\dif X_t=-\frac{\Omega}{2}~X_t\dif t+\Omega^{1/2}\dif W_t, 
$$
where $W_t$ is the one-dimensional standard Wiener process. 
\end{theorem}

The convergence of this process is a consequence of the multi-dimensional result established in Theorem \ref{theo:zz-main}. Figure~\ref{fig:zz-validation} shows that the theory agrees well with a numerical experiment.

\begin{figure}
\centering
\includegraphics[width=0.65\linewidth]{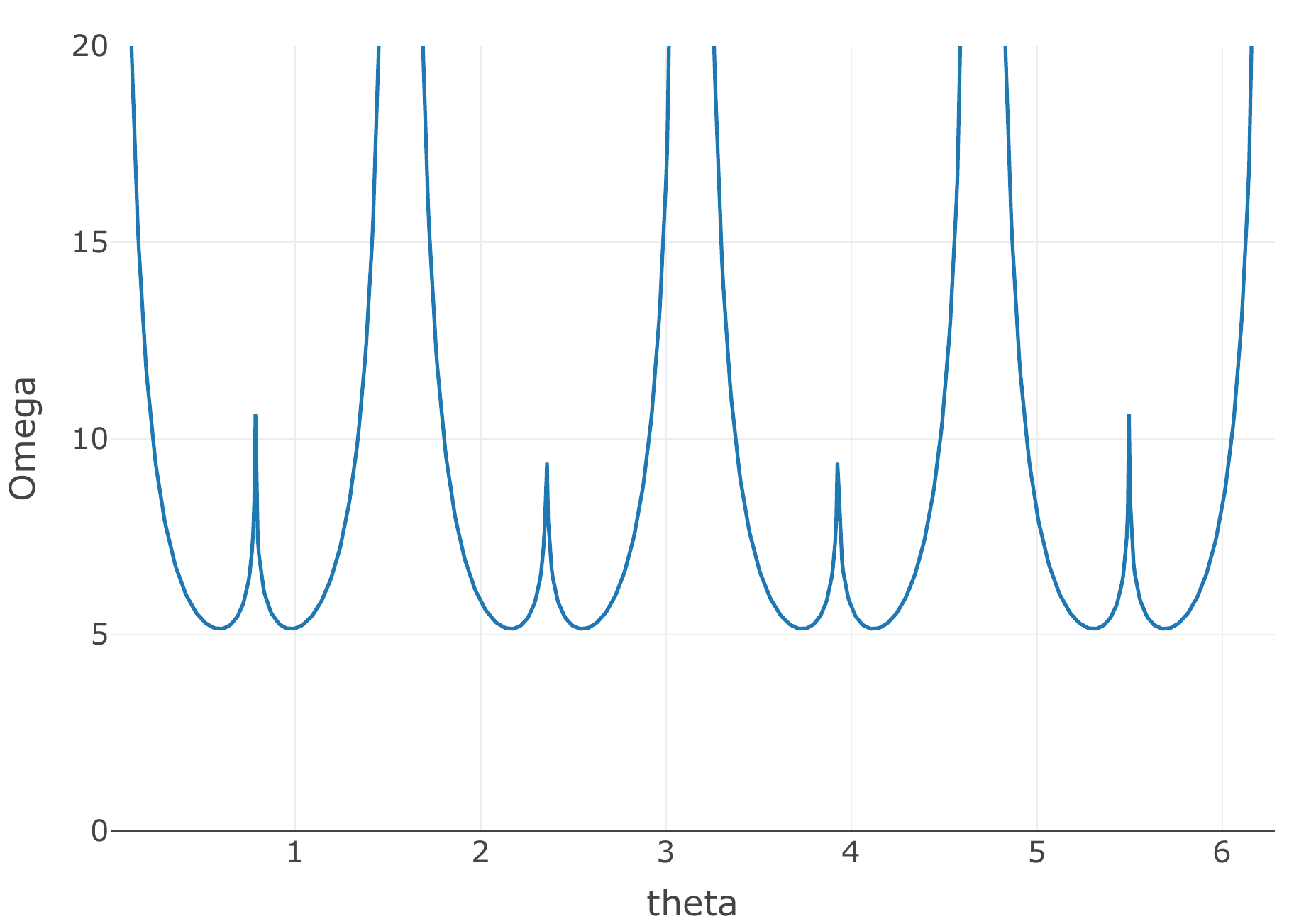}
\caption{The value of the diffusion coefficient $\Omega=\Omega(\theta)$ for $\theta\in [0,2\pi)$ for the Zig-Zag sampler. The function takes
$+\infty$ when $\theta=n\pi/4, n=0,1,\ldots, 7$ corresponding to the dotted lines.}
\label{fig:zz2defficiency}
\end{figure}

\begin{figure}[h]
\centering
\begin{subfigure}{0.48 \linewidth}
\includegraphics[width=\linewidth]{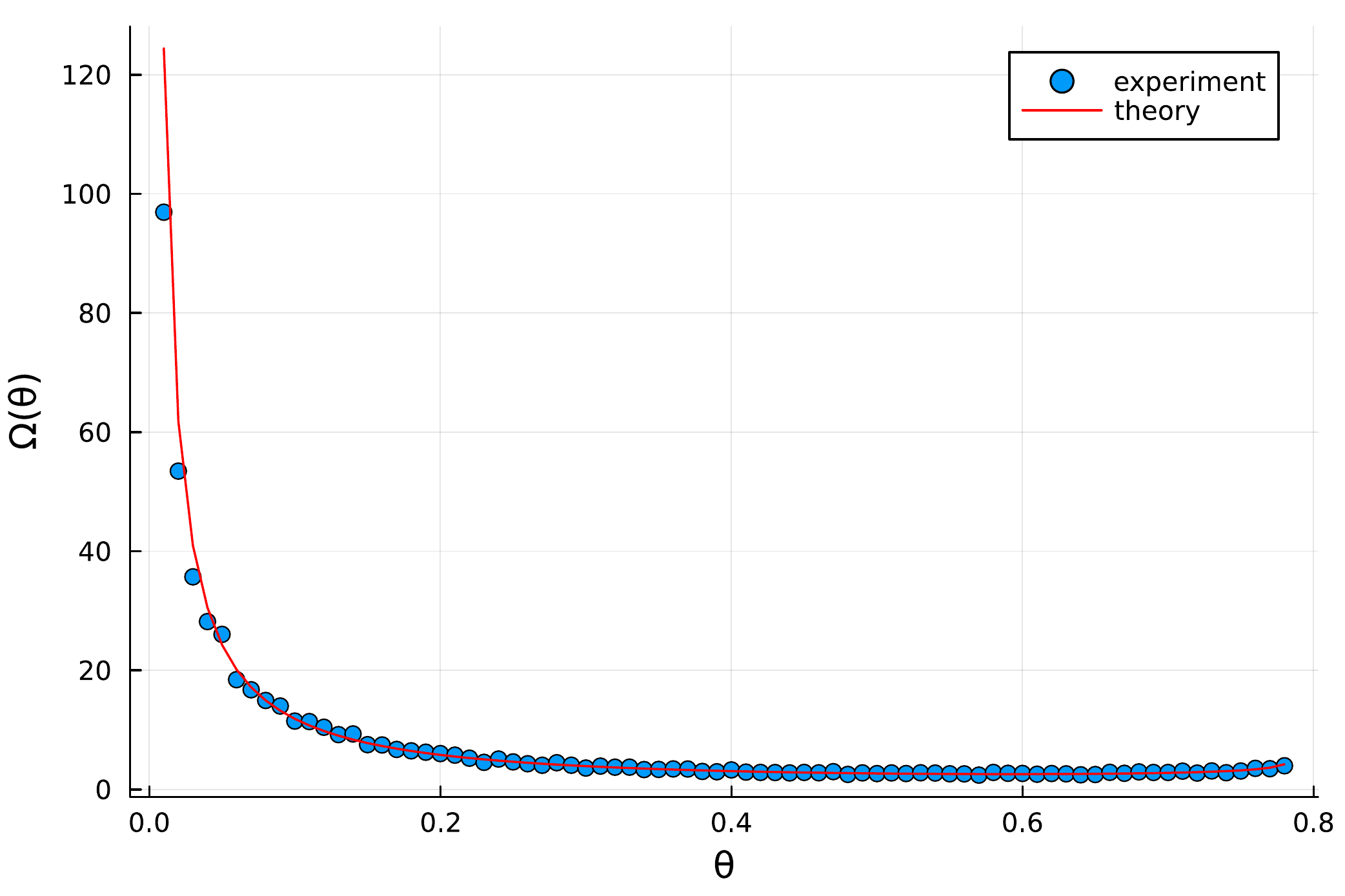}
\caption{$\theta \in (0,\pi/4)$}
\end{subfigure}
\quad 
\begin{subfigure}{0.48 \linewidth}
\includegraphics[width=\linewidth]{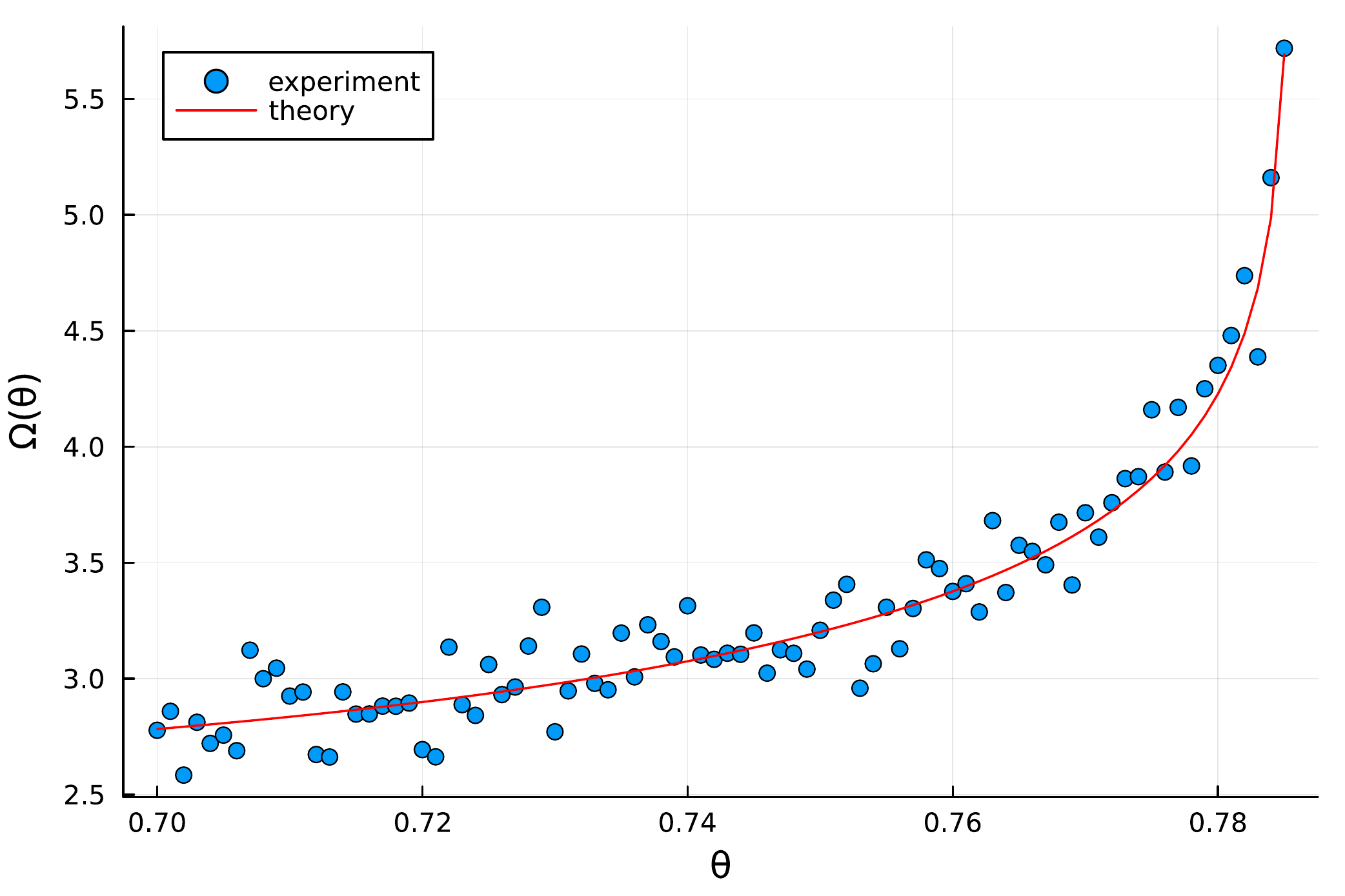}
\caption{Zoomed in at a sub-interval.}
\end{subfigure}
\caption{Numerical validation of the theoretical value $\Omega = \Omega(\theta)$ in the range $\theta \in (0,\pi/4)$. This is achieved by statistical estimation of the diffusion coefficient for a grid of values for $\theta$, interpreting the process $y_1^{\epsilon}(t)$ as a diffusion. This is for $\epsilon = 0.001$ and a trajectory with time horizon $T = 2.0$.}
\label{fig:zz-validation}
\end{figure}

\subsubsection{The fully and diagonally aligned cases of the Zig-Zag sampler}\label{sec:diagonally_aligned}



In this section, we will examine the special cases excluded in Assumption \ref{ass:zz2d}.
When $\theta$ is $\theta=n\pi/4\ (n=0,\ldots, 7)$, then $\Omega$ diverges. 
These eight cases can be further divided into two categories.
When $\theta=n\pi/4$ and $n$ is even, either $\cos\theta$ or $\sin\theta=0$ and as a result, $y_1$ and $y_2$ are independent. The path of $y_1$ in this scenario corresponds to the one-dimensional Zig-Zag sampler and exhibits good mixing behaviour without diffusive behaviour. We refer to these cases as the \emph{`fully aligned' scenario}. 

On the other hand, if $\theta$ equals $n\pi/4$ and $n$ is an odd number, the process remains diffusive but now allows for large jumps. These cases are referred to  as the \emph{`diagonally aligned' scenario}. 
We  give an informal but explicit description of how the trajectories behave in this case, and how this differs in comparison to the other cases. 

We will make use of the following simple  result.
\begin{lemma}
\label{thm:HalfGaussian}
Suppose $X$ is a positive random variable with hazard rate $(a+\gamma t)_+$ for some real number $a$ and positive constant $\gamma$, i.e., 
\[ \mathbb{P}( X \ge t) = \exp \left( -\int_0^t (a + \gamma s)_+ \, \dif s \right), \quad t \ge 0.\]
Then for any positive constant $c$,
$${\Bbb P}(X+a\gamma^{-1} \ge c \gamma ^{-1/2}) = \begin{cases} \exp\left(-c^2/2+a_+^2/2\gamma\right), \quad & c \ge a \gamma^{-1/2}, \\
1, \quad & 0 \le c < a \gamma^{-1/2},
\end{cases}
$$ and $$
{\Bbb E} (X+a\gamma^{-1}) = a_+/\gamma +\gamma^{-1/2}\sqrt{2\pi} ~\Phi(-a_+\gamma^{-1})~e^{a_+^2/2\gamma}.
$$
\end{lemma}


By~\eqref{eq:generator-zigzag-2d}, 
we have for a starting position $(y_1(0), y_2(0), v_1, v_2)$ that
\[ y_1(t) = y_1(0) + (v_1 \cos \theta -v_2 \sin \theta) t, \quad y_2(t) = y_2(0) + \epsilon^{-1}(v_1 \sin \theta + v_2 \cos \theta)t \]
until next jump. 
By the expression from the jump rates in~\eqref{eq:generator-zigzag-2d} 
the jump rates of the first coordinate and the second coordinate are summarised as
$$
\lambda_1(t) = (c_1+\epsilon^{-1}a_1+(1-\gamma_1) t+\epsilon^{-2}\gamma_1 t)_+,\quad 
\lambda_2(t) = (c_2+\epsilon^{-1}a_2+(1-\gamma_2) t+\epsilon^{-2}\gamma_2 t)_+,\ 
$$
where 
$$
\gamma_1=\sin\theta^2+v_1v_2\sin\theta\cos\theta,\quad 
\gamma_2=\cos\theta^2+v_1v_2\sin\theta\cos\theta
$$
and 
$$
a_1=y_2v_1\sin\theta,\quad
a_2=y_2v_2\cos\theta, 
$$
and $c_1, c_2\in\mathbb{R}$ are constants depending on the initial positions. 
Note that, if $v_1\sin\theta\neq -v_2\cos\theta$, then
$$
\gamma_1+\gamma_2=(v_1\sin\theta+v_2\cos\theta)^2>0\quad\Longrightarrow\quad \gamma_1>0\ \mathrm{or}\ \gamma_2>0. 
$$
On the other hand, $v_1\sin\theta= -v_2\cos\theta$ implies that $\gamma_1=\gamma_2=0$ and one of $a_1$ and $a_2$ is positive unless $y_2=0$. Note that, 
$|v_1\sin\theta|= |v_2\cos\theta|$
is satisfied if and only if $\theta$ is in the diagonally aligned cases. 
Moreover, in this case, the two events $v_1\sin\theta= -v_2\cos\theta$ and $v_1\sin\theta\neq -v_2\cos\theta$ occur one after the other and the events are interchanged by jumps.  

So by Lemma~\ref{thm:HalfGaussian}, we identify two cases for $(v_1,v_2)$
:
\begin{enumerate}
    \item If $v_1\sin\theta\neq -v_2\cos\theta$  then the next jump time is an $O_P(\epsilon)$ random variable;
    \item If $v_1\sin\theta= -v_2\cos\theta$  then  the next jump time is an $O_P(\epsilon~|y_2|^{-1})$ random variable.
\end{enumerate}
In fact it is easy to see that from starting values $(y_1(0),y_2(0))$ on $\mathbb{R}\times\mathbb{R}$ 
and with any initial velocities, the  switch random variables corresponding to $v_1$ and $v_2$ have explicit hazard rates, many of them in the form of that described in Lemma~\ref{thm:HalfGaussian}
with differing $\gamma $ values summarised in the following result. 

\begin{proposition}

\label{prp:hazards}
With the notation introduced above,
\begin{itemize}
    \item[(i)] 
    The next jump distribution is stochastically  bounded above by a random variable of the type appearing in Lemma~\ref{thm:HalfGaussian} with $$\gamma = \epsilon^{-2}\max\{\gamma_1,\gamma_2\}$$
    with $a=\max\{c_1,c_2\}+\epsilon^{-1}\max\{a_1, a_2\}$. 
    \item [(ii)]
    If $\theta \neq (2n+1)\pi/4$, then uniformly over the state space, as $\epsilon\to 0$ and the next jump distribution is bounded above by a $O_P(\epsilon)$ random variable.
    \item[(iii)]
    If $\theta = (2n+1)\pi/4$, then two subsequent jump distributions are of the order of $O(\epsilon~|y_2|^{-1})$ and $O_P(\epsilon)$ in turn. 
\end{itemize}
\end{proposition}

Thanks to these observations, and also by the simulation results, we conjecture that the process $y_1$ converges to a diffusion with jumps in this case. Here, the jump occurs when $|y_2|$ is small. However, this case is technically challenging. 
In a real-world application the narrow posterior distribution will likely not be perfectly aligned with the diagonal, and in an exceptional case of perfect diagonal alignment  we expect a better performance compared to our worst case analysis for the non-aligned case. 


\subsection{The Bouncy Particle Sampler in the two dimensional case}\label{subsec:2dbps}

\begin{figure}
\centering
\includegraphics[width=0.4\linewidth]{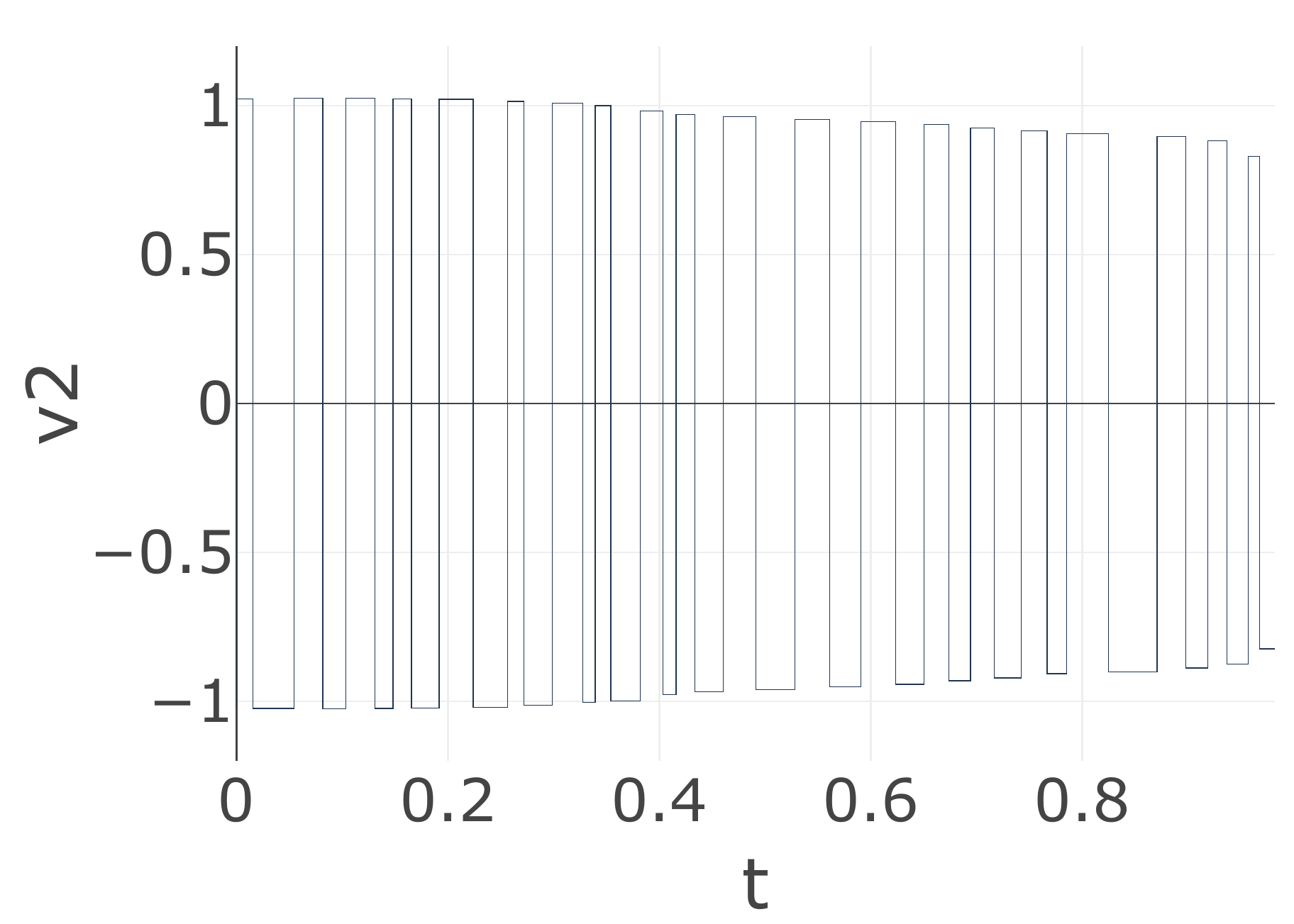}
\includegraphics[width=0.4\linewidth]{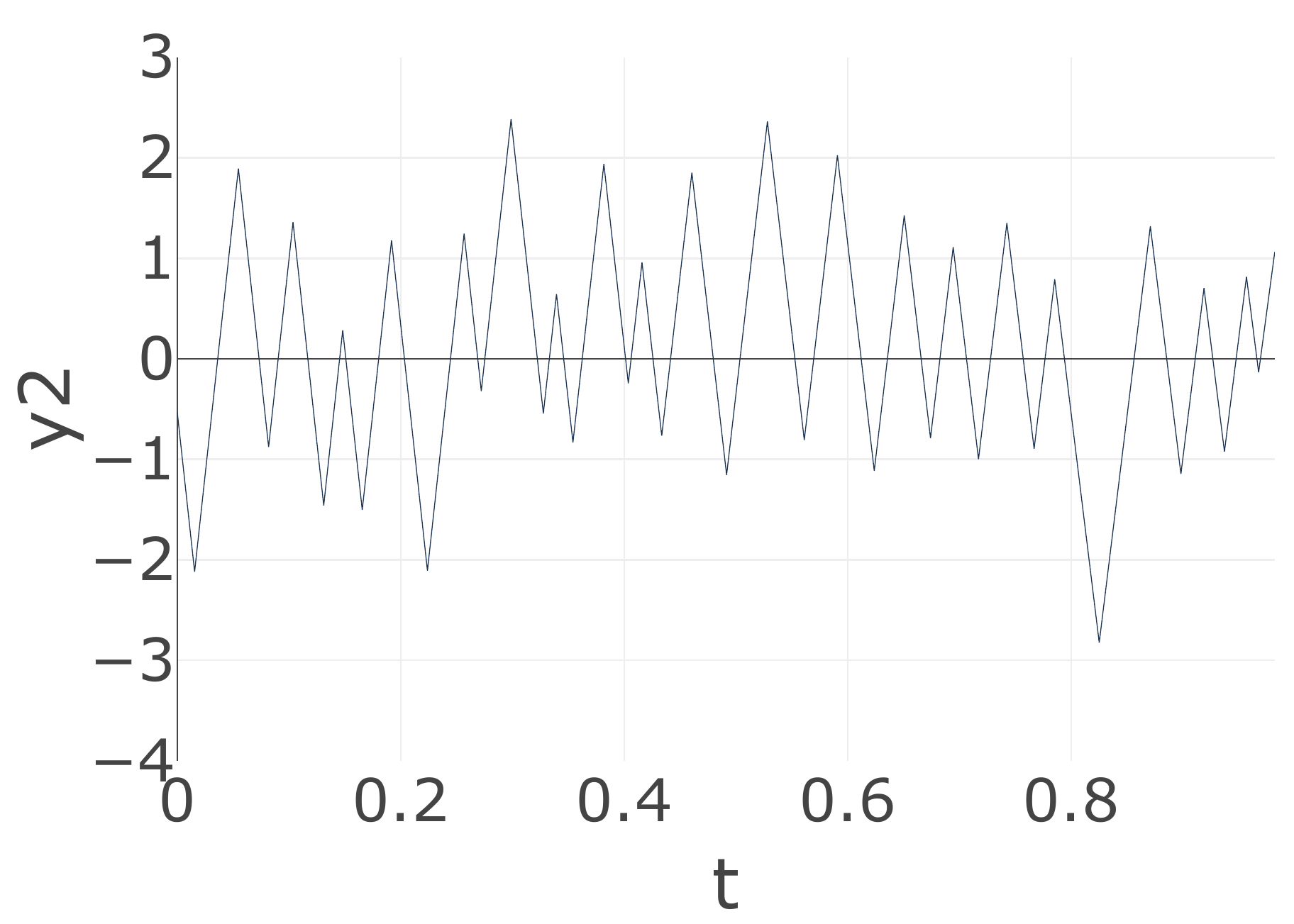}
\caption{Paths of 
$v_2$ (left) and 
$y_2$ (right) for the Bouncy Particle Sampler for $\epsilon=0.01$. 
 }
\label{fig:bps2d}
\end{figure}

Next, we present the asymptotic behaviour for the Bouncy Particle Sampler (BPS) for $d=2$.
The BPS is rotationally invariant and we may therefore assume $U=I_d$  without loss of generality.  The associated generator in terms of the coordinates $(x,v)$ is given by equation~\eqref{eq:generator-BPS}.
With the reparametrisation in (\ref{eq:reparametrise}), the generator can be expressed as 
\begin{align*}
    \mathcal{L}^\epsilon f(y,v)&=
        v_1\partial_{y_1}f(y,v) 
        +\epsilon^{-1}v_2\partial_{y_2}f(y,v)
        +(v_1y_1+\epsilon^{-1}v_2y_2)_+~(\mathcal{B}^\epsilon-\operatorname{{id}})f(y,v), 
\end{align*}
where we omit the refreshment term for simplicity (see Remarks~\ref{rem:refreshment} \added{and  \ref{rem:refreshment-ergodicity}} below), and where the operator $\mathcal{B}^\epsilon$ is 
given as the pullback of the reflection
$$
(y,v)\quad\mapsto\quad \left(y,\quad v-2\frac{v_1y_1+\epsilon^{-1}v_2y_2}{y_1^2+\epsilon^{-2}y_2^2}
\begin{pmatrix}
y_1\\
\epsilon^{-1}y_2
\end{pmatrix}\right)
=:\left(y,v+\begin{pmatrix}\Delta^\epsilon v_1\\ \Delta^\epsilon v_2\end{pmatrix}\right)
. 
$$
Observe that the jump size $\Delta^\epsilon v_1$ of $v_1$ induced by $\mathcal{B}^\epsilon$ satisfies
\begin{align}
\epsilon^{-1}\Delta^\epsilon v_1\longrightarrow_{\epsilon\rightarrow 0}\quad -2\frac{v_2~y_2}{y_2^2}~y_1. 
\label{eq:bps-short-expansion-v1}
\end{align}
Our first observation is that, as shown in Figure \ref{fig:bps2d}, the process of  $\mathrm{sgn}(v_2)$ and that of $y_2$ mixes well but that of $(y_1,v_1)$ does not. This is not surprising given the form of the generator. When $\epsilon$ is small enough, the process of $(y_2,v_2)$ is asymptotically dominated by the generator $\mathcal{L}_0$ corresponding to the one-dimensional Zig-Zag sampler
\begin{align*}
    \mathcal{L}_0 f(y,v)&=
        v_2\partial_{y_2}f(y,v)
        +(v_2y_2)_+~(\mathcal{B}^0-\operatorname{{id}})f(y,v), 
\end{align*}
where $\mathcal{B}^0$
is given as the pullback of the coordinate reflection
$$
(x,v)\mapsto\left(x,\begin{pmatrix}
v_1\\-v_2
\end{pmatrix}\right). 
$$
Note that the value of $|v_2|$ remains unchanged by the dynamics generated by $\mathcal{L}_0$.

The second observation is that, as in Figure \ref{fig:bps-numerical}, the process $(y_1, v_1)$ (and $|v_2|$) mixes slowly and vaguely follows a periodic orbit. 
The process of $y_1$ satisfies $\frac{\dif}{\dif t} y_1 = v_1$. Although the process of  $v_1$ is driven by the reflection jump, the jump size becomes so small that  it asymptotically follows  
$$
\frac{\dif }{\dif t} v_1=(v_2y_2)_+~\left(-2\frac{v_2y_2}{y_2^2}y_1\right)=-2v_2^2~1_{\{v_2y_2>0\}}y_1
$$
as $\epsilon \rightarrow 0$ by (\ref{eq:bps-short-expansion-v1}). 
Observe that the value of $|v|=\sqrt{v_1^2+v_2^2}$ does not change during the dynamics. Thus, for $\kappa^2:=|v(0)|^2$,  we have $|v_2|=\sqrt{\kappa^2-v_1^2}$. 
Also, the fast mixing component $( y_2, \mathrm{sgn}(v_2))$ is averaged out by  $\mathcal{N}(0,1)\otimes\mathcal{U}(\{-1,+1\})$, and we obtain the dynamics
\begin{equation}
\label{eq:bps-2dlimit}
  \frac{\dif}{\dif t}  \begin{pmatrix}
    y_1\\
    v_1
    \end{pmatrix}
    =
    \begin{pmatrix}
    v_1\\
    -(\kappa^2-v_1^2)y_1
    \end{pmatrix}. 
\end{equation}
See the red lines of Figure \ref{fig:bps-numerical}. 


\begin{figure}
    \centering
    \begin{subfigure}{0.48 \textwidth}
    \includegraphics[width=\textwidth]{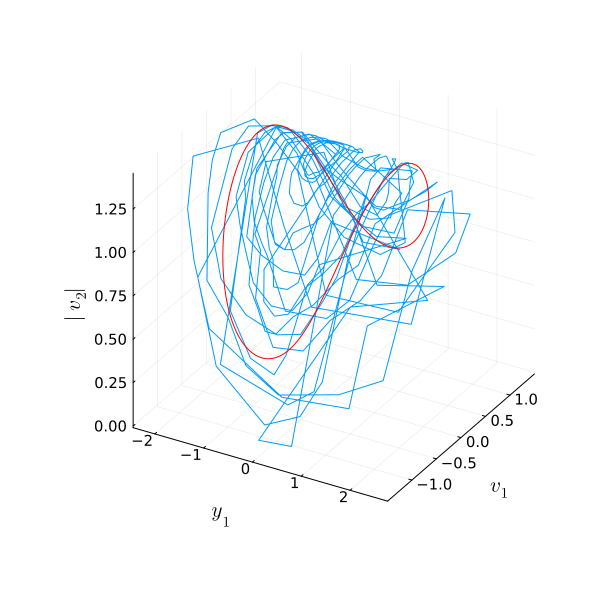}
    \caption{$\epsilon = 0.1$}
    \end{subfigure}\quad\quad
    \begin{subfigure}{0.48 \textwidth}
    \includegraphics[width=\textwidth]{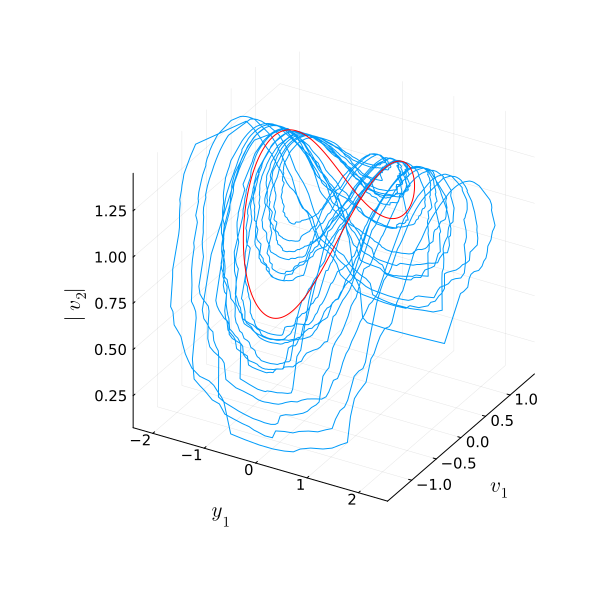}
    \caption{$\epsilon = 0.01$}
    \end{subfigure}     \\
    \begin{subfigure}{0.48 \textwidth}
    \includegraphics[width=0.8\textwidth]{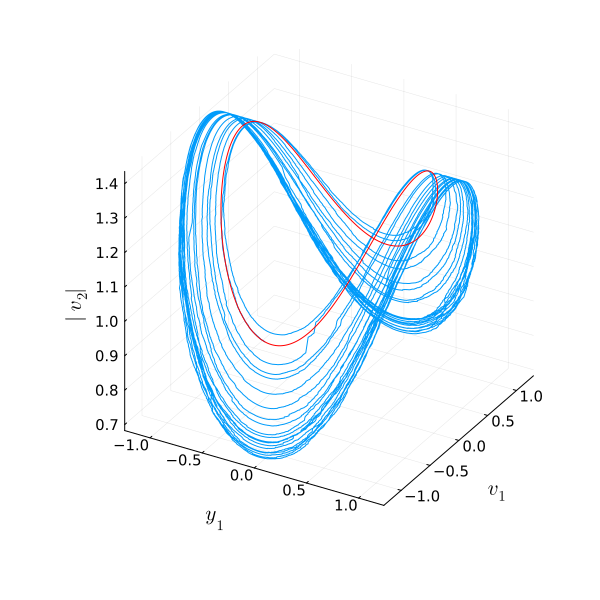}
    \caption{$\epsilon = 10^{-4}$}
    \end{subfigure}    \begin{subfigure}{0.48 \textwidth}
    \includegraphics[width=\textwidth]{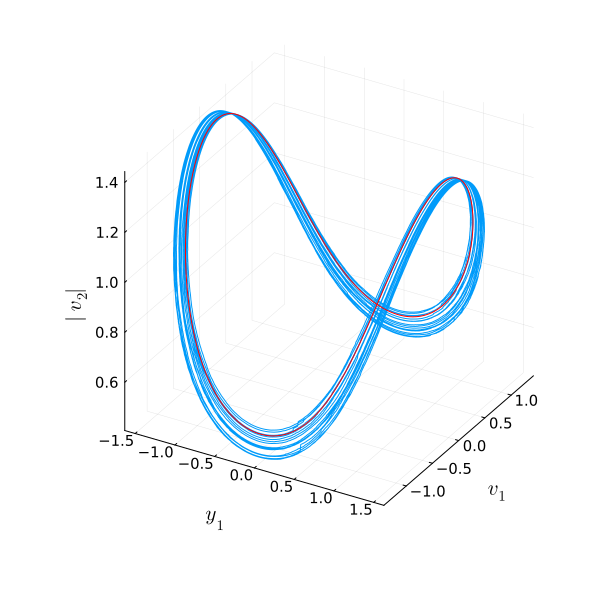}
    \caption{$\epsilon = 10^{-5}$}
    \end{subfigure}
    \caption{Numerical experiment illustrating how the $(y_1(t), v_1(t), |v_2|(t))$ statistics of a BPS trajectory (in blue) converge to the fluid limit (in red) given by~\eqref{eq:bps-2dlimit}.}
    \label{fig:bps-numerical}
\end{figure}

 \deleted{This dynamics preserves the value of $y_1^2+\log(|v(0)|^2-v_1^2)$.} Let $(y_1^\epsilon(t),v_1^\epsilon(t))$ be the $(y_1,v_1)$-coordinate of the Bouncy Particle Sampler process. We have the following convergence result. 

\begin{theorem}
The stationary process $(y_1^\epsilon(t), v_1^\epsilon(t))$ converges to the solution of the differential equation (\ref{eq:bps-2dlimit}) 
in the sense that 
$$
\sup_{0\le t\le T}|y_1^\epsilon(t)-y_1(t)|+|v_1^\epsilon(t)-v_1(t)|\longrightarrow_{\epsilon\rightarrow 0} 0
$$
in probability for any $T>0$ 
where $(y_1^\epsilon(0),v_1^\epsilon(0))=(y_1(0),v_1(0))$, $\kappa^2=|v^\epsilon(0)|^2$. 
\end{theorem}

\begin{proof}
The claim follows from the general result Theorem \ref{theo:bps-main} 
replacing $c(\alpha,\beta)$ by $\mathbb{E}[(v_2 y_2)_+^2/|y_2|^2]=1/2$ since the Skorokhod topology is equivalent to the local uniform topology when the limit process is continuous as described in page 124 of  Section 12 of \cite{MR1700749}. 
\end{proof}

\begin{remark}
\label{rem:refreshment}
For simplicity we have assumed the refreshment to be equal to zero. The only thing that would change for a fixed refreshment rate would be that the process $(y_1,y_2,v_1,v_2)$ would experience refreshments at $O_P(1)$ random times. The trajectory between refreshments would not be affected by these refreshments, aside from having a new initial velocity vector. For the multi-dimensional case with refreshment jumps, refer to Theorem \ref{theo:bps-main}. 
\end{remark}

\begin{remark}
\label{rem:refreshment-ergodicity}
\added{
In the absence of refreshment, the limit process retains the value of $y_1^2 + \log(|v(0)|^2 - v_1^2)$, suggesting that the process is non-ergodic. However, with the introduction of refreshment, proving the process's ergodicity becomes straightforward. As this is beyond the scope of this paper, we do not include the details.
}
\end{remark}

\subsection{Discussion of the two-dimensional case}
\label{sec:2dim-discussion}

\added{We summarise the results of the two-dimensional case in Table \ref{tab:computational-effort}. Here, we estimate the time until convergence using a factor that accelerates convergence to a limiting process, which exhibits mixing properties. Additionally, we estimate the computational cost per unit time of the processes based on the number of jump events per unit time (see the discussion at the end of the section). The combined effort is then calculated as the product of these two factors.}

Both the Zig-Zag Sampler (ZZS) and the Bouncy Particle Sampler (BPS) are capable of simulating anisotropic target distribution. 
However, the asymptotic behaviour of the two methods is systematically different. The former converges to a diffusion process and the latter to a deterministic periodic orbit. Moreover, the number of jumps to convergence is $O(\epsilon^{-2})$ for the ZZS (except for the non-generic, aligned cases), and $O(\epsilon^{-1})$ for the BPS. Therefore, the BPS is efficient in this scenario, which is in stark contrast to the  high-dimensional, factorised  scenario  studied in \cite{bierkens2018high}. 

The behaviour of the Zig-Zag sampler is divided into three categories: Fully aligned \added{($\theta= n\pi/4$, $n$: even)}, diagonally-aligned \added{($\theta= n\pi/4$, $n$: odd)}, and non-aligned \added{($\theta\neq n\pi/4$)}. We are not yet fully able to understand the asymptotic behaviour of the diagonally aligned case. However, we conjecture that it converges to a diffusion with jumps. 

For a discussion on a multi-dimensional scenario and a comparison of the computational complexity between PDMPs and the random walk Metropolis algorithm, please refer to Section \ref{sect:disc}.


\begin{table}[ht!]
\begin{center}
\begin{tabular}{ l  c  c  c}
  Method & $\sharp$ Events/unit time & Time until convergence & Combined effort \\
  \hline			
  ZZ ($\theta\neq n\pi/4$)& $O(\epsilon^{-1})$ & $O(\epsilon^{-1})$ & $O(\epsilon^{-2})$ \\
  ZZ ($\theta= n\pi/\replaced{2}{4}$)& $O(\epsilon^{-1})$ & $O(1)$ & $O(\epsilon^{-1})$ \\
  BPS & $O(\epsilon^{-1})$  & $O(1)$ & $O(\epsilon^{-1})$\\
\end{tabular}
\caption{Computational effort of the piecewise deterministic processes. \added{The Zig-Zag sampler in the case of diagonally-aligned is not presented in this table.}}
\label{tab:computational-effort}
\end{center}
\end{table}
 
\added{Note here that estimating the computational cost is challenging since the implementation methods for PDMPs are not as established as those for Markov chain Monte Carlo methods. We assume the use of the Poisson thinning strategy to simulate jump times. This strategy employs a starting Poisson process with a higher conditional intensity function than that of the process we aim to simulate. We pick the true jump times from jump times of the starting Poisson process by the ratio of the conditional intensity functions. }

\added{
Simulating the starting Poisson process is assumed to be straightforward, whereas calculating the ratio of the conditional intensities is computationally demanding due to the need to evaluate the derivative of the log of the negative density of the target distribution. Therefore, it is reasonable to estimate that the computational cost is linearly related to the number of jump events in the starting Poisson process. We assume that the number of jumps in the starting process is estimated by a constant multiple of that in the target process, and this multiplier is consistent across all processes listed in the table. Under this assumption, the computational effort per unit time can be estimated by the number of jump events in the target process. 
}

\added{
It is important to note, however, that this is just an assumption and other estimates could also be valid depending on different factors and assumptions, including variability in implementation techniques or computational capabilities.
}

\subsection{\added{Non-Gaussian distributions}}

\added{Our proof strategy assumes that the target distribution is multivariate Gaussian. In this section we investigate whether our results are seen to hold empirically beyond the Gaussian regime.}

\added{In Figure~\ref{fig:non-gaussian-experiment} we establish numerically the dependence of the estimated asymptotic variance (with respect to the second moment) as a function of the anisotropy $1/\epsilon$, both for a Gaussian distribution with covariance matrix $\Sigma^{\epsilon}$ and for a heavy-tailed multivariate Student distribution with parameters $\nu =4$ (degrees of freedom) and covariance matrix $\Sigma^{\epsilon}$.}

\added{Here the asymptotic variance of a Markov process $(X_t^{\epsilon})$ satisfying a CLT with respect to a function $f$ is defined as
\[ \sigma^2_{\infty, \epsilon}(f) = \lim_{T \rightarrow \infty} \operatorname{Var} \left(\frac 1 {\sqrt T} \int_0^T f(X^{\epsilon}_s) \, ds  \right).\]
It may be checked that, if $(X_{t/\epsilon^{\beta}}^{\epsilon}) \Rightarrow (X_t^{\infty})$ for a limiting process $X^{\infty}$ as $\epsilon \downarrow 0$, along with certain additional conditions regarding the convergence of probability laws and ergodic properties, then $\sigma_{\infty,\epsilon}^2 \sim \epsilon^{-\beta} \sigma^2_{\infty}$, where $\sigma^2_{\infty}$ is the asymptotic variance of $(X_t^{\infty}).$ Therefore, the estimated asymptotic variance may be used to estimate the asymptotic speed decrease as influenced by $\epsilon$. The slopes of the regression lines plotted in these figures are displayed in Table~\ref{tab:regression-slopes-non-gaussian} and give an empirical measure of the appropriate time scaling of the process with varying anisotropy.}

\added{Our implementation for the experiment is efficient in the sense that it leads to an approximate 0.4 acceptance probability of proposed switches for all values of $\epsilon$ for the Student distribution, illustrating that in this experiment the computational efficiency of the process does not deterioriate with increasing anisotropy, as argued in Subsection~\ref{sec:2dim-discussion}. }

\begin{figure}
\begin{subfigure}{0.45 \textwidth}
\includegraphics[width=\textwidth]{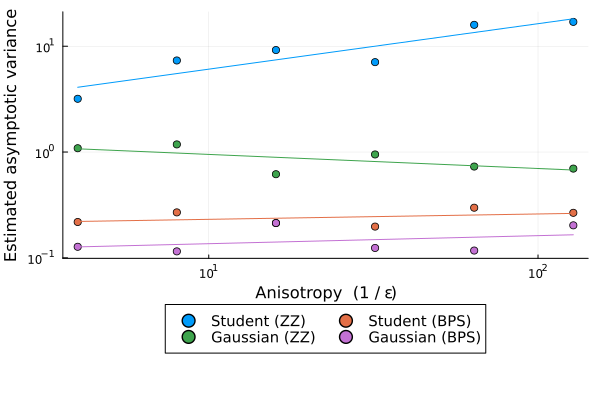}
\caption{$\theta = 0$}
\end{subfigure}
\hfill 
\begin{subfigure}{0.45 \textwidth}
\includegraphics[width=\textwidth]{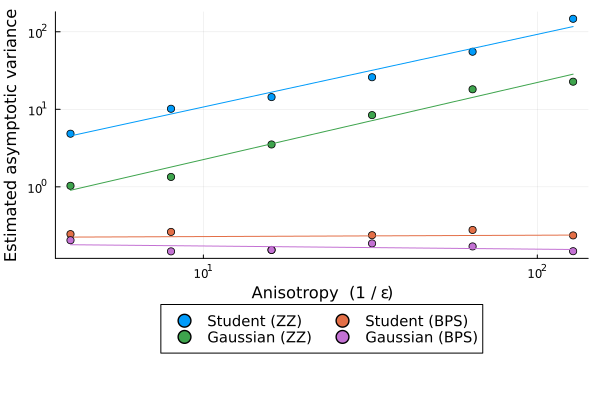}
\caption{$\theta = \pi/6$}
\end{subfigure}
\\
\begin{subfigure}{0.45 \textwidth}
\includegraphics[width=\textwidth]{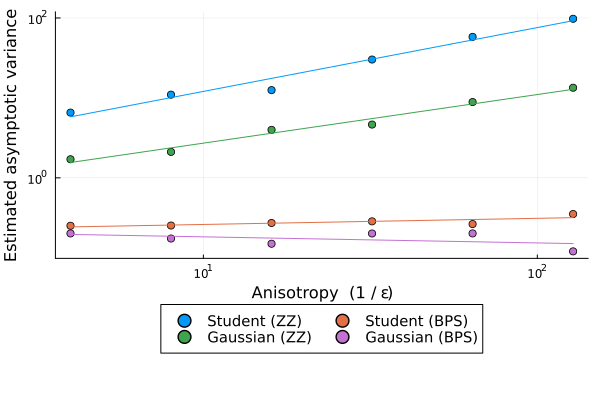}
\caption{$\theta = \pi/4$}
\end{subfigure}
\hfill 
\begin{subfigure}{0.45 \textwidth}
\includegraphics[width=\textwidth]{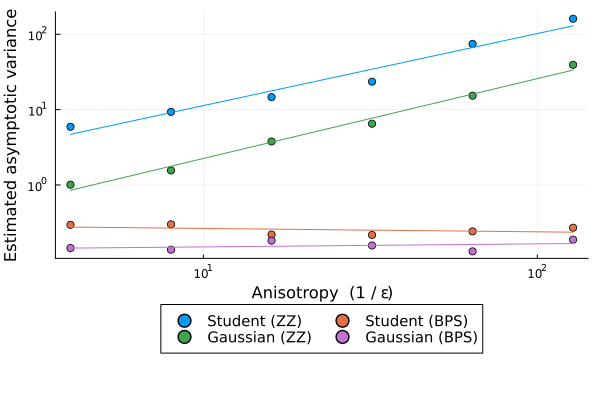}
\caption{$\theta = \pi/3$}
\end{subfigure}
\caption{Log-log plot of the dependence of the estimated asymptotic variance on $1/\epsilon$ for a multivariate Gaussian distribution and a multivariate student distribution with $\nu = 4$ degrees of freedom. Plotted are the (geometric) mean over $100$ experiments of the asymptotic variance, for simulations with a time horizon $T = 10/\epsilon$. The time horizon is made to vary with $\epsilon$ in order to estimate asymptotic variance reliably.}
\label{fig:non-gaussian-experiment}
\end{figure}

\begin{table}[ht!]
\begin{tabular}{l|c|c|c|c}
$\theta$ & 0 & $\pi/6$ & $ \pi/4$ & $\pi/3$ \\
\hline
Gaussian (ZZ) & -0.13 $\pm$ 0.07 & 1.00 $\pm$ 0.07 & 0.61 $\pm$ 0.08 & 1.06 $\pm$ 0.07 \\
Gaussian (BPS) &  0.07 $\pm$ 0.07 & -0.04 $\pm$ 0.07 & -0.08 $\pm$ 0.06 & 0.04 $\pm$ 0.07 \\
Student (ZZ) & 0.43 $\pm$ 0.07 & 0.94 $\pm$ 0.07& 0.80 $\pm$ 0.07 & 0.96 $\pm$ 0.07\\
Student (BPS) & 0.05 $\pm$ 0.07 & 0.02 $\pm$ 0.06 & 0.08 $\pm$ 0.06 & .05 $\pm$ 0.06
\end{tabular}

\caption{Regression slopes corresponding to Figure~\ref{fig:non-gaussian-experiment}, indicating the  dependence of the logarithm of the asymptotic variance on $\log (1/\epsilon)$. An estimated slope of $\beta$ indicates scaling of the asymptotic variance by a factor $1/\epsilon^{\beta}$. The theory of this paper predicts $\beta \approx 1$ for Zig-Zag, and $\beta \approx 0$ for BPS, in the Gaussian case. This table and Figure~\ref{fig:non-gaussian-experiment} suggest that these observations may also apply beyond the Gaussian regime.}
\label{tab:regression-slopes-non-gaussian}
\end{table}

\section{Main results and proof strategy}
\label{sect:main}

Let $d\ge 2$ be an integer, and let $\former$ and $\latter$ be a partition  of $\{1,\ldots, d\}$ that is, 
\begin{align*}
    \former\cup\latter=\{1,\ldots,d\},\quad \former\cap\latter=\emptyset 
\end{align*}
and  $0<\sharp\former=k$ and $0<\sharp\latter=l=d-k$. 
The target probability distribution is  $\mathcal{N}_d(0,\Sigma^\epsilon)$, where the $d\times d$-symmetric positive definite matrix $\Sigma^\epsilon$ has the eigen decomposition
$$
\Sigma^\epsilon=U^{\top}(\Lambda^\epsilon)^2 U,\ \Lambda^\epsilon=\operatorname{diag}(\Lambda_\former, \epsilon~\Lambda_\latter), 
$$
where $\Lambda_\former$ and $\Lambda_\latter$ are $k\times k$ and $l\times l$ diagonal matrices with strictly positive diagonal elements, and $U$ is an orthogonal matrix with decomposition
$$
U=
\added{
    \begin{pmatrix}
    U_{\former\former}&
    U_{\former\latter}\\
    U_{\latter\former}&
    U_{\latter\latter}
\end{pmatrix}
}. 
$$
In this paper, we will study two kinds of extended generators, the Zig-Zag sampler defined on $\mathbb{R}^d\times\{-1,+1\}^d$ and the Bouncy Particle Sampler defined on $\mathbb{R}^d\times\mathbb{R}^d$: 
\begin{align}
\label{eq:generator-zigzag} \text{ZZS}:\quad\mathfrak{L}^{\epsilon} f(x,v) &= v^{\top}\partial_xf(x,v)  + \sum_{i=1}^d (v_i ((\Sigma^\epsilon)^{-1} x)_i)_+ (\mathcal{F}_i-\operatorname{id})f(x,v)\\
\label{eq:generator-BPS}\text{BPS}:\quad\mathfrak{L}^\epsilon f(x,v)&=
        v^{\top}\partial_xf(x,v)
        +(v^{\top}(\Sigma^\epsilon)^{-1}x)_+~(\mathfrak{B}^\epsilon-\operatorname{id})f(x,v)
        +\rho~\left(\mathcal{R}-\operatorname{id}\right)f(x,v), 
\end{align}
where the reflection operator  $\mathfrak{B}^\epsilon$ is a deterministic operation such that 
\begin{align*}
    (x,v)\quad\mapsto\quad \left(x,\quad v-2 v^{\top}(\Sigma^\epsilon)^{-1}x~\frac{(\Sigma^\epsilon)^{-1}x}{|(\Sigma^\epsilon)^{-1}x|^2}\right). 
\end{align*} 
The constant $\rho\ge 0$ is called the refreshment rate, and the refresh operator $\mathcal{R}$ is defined by 
$$
(\mathcal{R}f)(x,v)=\int_{\mathbb{R}^d}f(x,w)\phi(w)\dif w
$$
where $\phi(w)$ is the probability density function of the $d$-dimensional standard normal distribution. The domains of the extended generators are the set of bounded functions on $(x,v)\in \mathbb{R}^d\times\{-1,+1\}^d$ and $(x,v)\in \mathbb{R}^d\times\mathbb{R}^d$  such that each function is differentiable in $x$. The cores of the extended generators have also been studied such as \cite{Durmus_2021, 10.1214/21-ECP430}.

\subsection{Main results for the Zig-Zag sampler}
\label{subsec:zigzag}

With the reparametrisation (\ref{eq:reparametrise}), the extended generator of the  Zig-Zag process $\xi^\epsilon(t)=(y^\epsilon(t), v^\epsilon(t))$  is
\begin{equation}
\label{eq:zig-zag-operator}
    \mathcal{L}^\epsilon f(y,v)=((\Lambda^\epsilon)^{-1} Uv)^{\top}\partial_y 
    f(y,v)+\sum_{i=1}^d (v_i(U^{\top}(\Lambda^\epsilon)^{-1} y)_i)_+~(\mathcal{F}_i-\operatorname{id})f(y,v). 
\end{equation}
We decompose the $d$-dimensional vector $y$ into $\former$ and $\latter$ components, i.e., $y=(y_\former, y_\latter)$. 
When $\epsilon$ is small, except for some special cases, the process of  $(y_\latter, v)$ mixes well but that of $y_\former$ does not. We are interested in the latter behaviour since the slowest dynamics reflects the total computational cost of Monte Carlo methods. The limit process of the faster dynamics $(y_\latter, v)$ is characterised by the extended generator 
\begin{equation}
\label{eq:zig-zag-operator-fast}
    \mathcal{L}_0=(\Lambda_\latter^{-1}(Uv)_\latter)^{\top}\partial_{y_\latter}
    +\sum_{i=1}^d (v_i(U_{_\latter,\cdot}^{\top}\Lambda_{\latter}^{-1}y_\latter)_i)_+~(\mathcal{F}_i-\operatorname{id}), 
\end{equation}
\added{where, $U_{\latter,\cdot}=(U_{\latter\former}, U_{\latter\latter})$. }
Before discussing the asymptotic result, we describe the special cases mentioned above. 
In the two-dimensional case described in Section \ref{subsec:zz-digest},  $\theta=n\pi/4\ (n=0,\ldots, 7)$ were excluded in the analysis, which corresponds to the case where the Markov process with generator $\mathcal{L}_0$ is not ergodic as a process of $(y_2, v)$. Ergodicity is related to the $d\times d$-positive semidefinite matrix
$$
\Theta_\latter=U_{\latter, \cdot}^{\top}\Lambda_\latter^{-2}U_{\latter,\cdot}. 
$$
If $v_i(\Theta_\latter v)_i>0$, then the $i$-th component $v_i$ is flippable with  direction $v$, in the sense that if the process maintains  direction $v$,  there is a positive probability of switching the sign of the $i$-th component of $v$ when $t$ is large enough. 
Indeed, 
$$
(v_i(t)(U_{_\latter,\cdot}^{\top}\Lambda_{\latter}^{-1}y_\latter(t))_i)=(v_i(U_{_\latter,\cdot}^{\top}\Lambda_{\latter}^{-1}y_\latter(0))_i+t~v_i(\Theta_\latter v)_i)>0
$$
when $t$ is large enough, where 
$v(t)=(v_1(t),\ldots, v_d(t))=v$, and $y_\latter (t)= y_\latter(0)+t~\Lambda_\latter^{-1}(Uv)_\latter$. 
We assume the following condition for flippability, which ensures the ergodicity of the process. 

\begin{assumption} 
\label{as:zz-condition}
$v^{\top}\Theta_\latter v> 0$ for any $v\in\{-1,0,+1\}^d\backslash\mathbf{0}$. 
\end{assumption}

As described above, the process $(y^\epsilon_\latter(t), v^\epsilon(t))$ converges weakly to a Markov process with an extended generator $\mathcal{L}_0$. However, in the same time scale, the process $y_\former^\epsilon(t)$ becomes degenerate. To address this, we must scale time appropriately.  \replaced{Specifically, we estimate \(O(\epsilon^{-1})\) time for mixing. This is because the time for the process \(y_{\former}^\epsilon\) to converge is accelerated by a factor of \(n\), and the limiting Ornstein--Uhlenbeck process exhibits strong mixing properties.}{Specifically, we require $O(\epsilon^{-1})$ time for mixing.}  The proof will be presented at the end of Section \suppA. 

\begin{theorem}
\label{theo:zz-main}
As $\epsilon\rightarrow 0$, under Assumption \ref{as:zz-condition}, and if $\xi^\epsilon(0)$ follows the stationary distribution $\mathcal{N}(0, I_d)\otimes\mathcal{U}(\{-1,+1\}^d)$, then $y_\former^\epsilon(\epsilon^{-1}t)$ converges to an Ornstein--Uhlenbeck process defined in (\ref{eq:zz-ou})
in the Skorokhod topology. 
\end{theorem}

Under the same conditions as in Theorem \ref{theo:zz-main}, we examine the quantity $N_T^\epsilon$ which represents the number of jumps in the process $(y^\epsilon(t),v^\epsilon(t))$ over the time interval \deleted{of} $[0,T]$.  The number of jumps, $N_T^\epsilon$, is an indicator of the computational cost per unit time.  Let $\Sigma^\epsilon_{i,i}$ and $(\Theta_\latter)_{i,i}$ represent the $i$-th diagonal entries of the matrices $\Sigma^\epsilon$ and $\Theta_\latter$, respectively. The proof is postponed to Section \suppA.  

Based on the following theorem, we can conclude that $\mathbb{E}[N_T^\epsilon]=O(\epsilon^{-1} T)$. This implies that $O(\epsilon^{-2})$ jumps are required to achieve mixing, as we make precise in the following result. 

\begin{theorem}
\label{theo:zz-costs}
\added{Let $\theta_{i,i}^\epsilon$ be the   $(i,i)$-th element of $(\Sigma^\epsilon)^{-1}$ for $i=1,\ldots, d$. }
Under the same conditions as Theorem \ref{theo:zz-main}, the expected number of jumps of the process $(y^\epsilon(t),v^\epsilon(t))$ in the time interval $[0,T]$ is
\begin{align*}
\mathbb{E}[N_T^\epsilon]=\frac{T}{2\sqrt{2\pi}}\sum_{i=1}^d\added{\sqrt{\theta_{i,i}^\epsilon}}.
\end{align*}
In particular, 
\begin{align*}
\lim_{\epsilon\rightarrow 0}\epsilon~\mathbb{E}[N_T^\epsilon]=\frac{T}{2\sqrt{2\pi}}\sum_{i=1}^d\sqrt{(\Theta_\latter)_{i,i}}. 
\end{align*}
\end{theorem}

\subsection{Main results for the Bouncy Particle Sampler}
\label{subsec:bps}

With the reparametrisation (\ref{eq:reparametrise}), 
the extended generator of the Bouncy Particle Sampler process $\xi^\epsilon(t)=(y^\epsilon(t),v^\epsilon(t))$ is 
\begin{equation}
\label{eq:bps-operator}
    \mathcal{L}^\epsilon=
        ((\Lambda^\epsilon)^{-1}v)^{\top}\partial_y 
        +(v^{\top}(\Lambda^\epsilon)^{-1}y)_+~(\mathcal{B}^\epsilon-\operatorname{{id}})
        +\rho~\left(\mathcal{R}-\operatorname{id}\right), 
\end{equation}
where the reflection operator is defined by $\mathcal{B}^\epsilon f(y,v)=f(y,B^\epsilon(y)v)$ and 
\begin{align}
\label{eq:bps_bounce}
    B^\epsilon(y)v=v-2 v^{\top}(\Lambda^\epsilon)^{-1}y~\frac{(\Lambda^\epsilon)^{-1}y}{|(\Lambda^\epsilon)^{-1}y|^2}.
\end{align} 
Unlike Zig-Zag sampler case, without loss of generality, we can assume $\Sigma^\epsilon=(\Lambda^\epsilon)^2$ since the process is rotationally invariant. 
Let $y^\epsilon(t)=(y^\epsilon_{\former}(t), y^\epsilon_{\latter}(t))$ and 
$v^\epsilon(t)=(v^\epsilon_{\former}(t), v^\epsilon_{\latter}(t))$. In two dimensional case in  Section \ref{subsec:2dbps}, $(y_2,v_2)$ mixes faster and $(y_1,v_1)$ mixes slower.   Similarly, in the general case, 
$(y_\latter, v_\latter)$ mixes faster except some statistics described below,  and $(y_\former, v_\former)$ mixes slower. 

The  limit behaviour depends on the value of $\Lambda_\latter$. \added{Using a different matrix would result in a different limit process, although the rate would remain unchanged.} For the sake of simplicity,  we assume $\Lambda_\latter$ is the identity matrix. In this case, 
 $\alpha=|v_\latter|^2$, and $\beta=|v_\latter|^2|y_\latter|^2-(v_\latter^{\top}y_\latter)^2$ are fixed constants, and $(y_\latter, v_\latter)$ is always on a two dimensional hyperplane. This hyperplane will be refreshed at the refreshment jump time. In the asymptotic analysis, we need to study not only $(y_\former,v_\former)$ but also $(\alpha,\beta)$ together with the effect of refreshment. 
 The limit process is a little complicated, and so we will describe it Section \suppB. The proof of the following theorem is also in the section. 
Let $(\alpha^\epsilon(t),\beta^\epsilon(t))$ be the corresponding process of the statistics $(\alpha,\beta)$.

\begin{theorem}
\label{theo:bps-main}
As $\epsilon\rightarrow 0$, if $\Lambda_\latter= I_\latter$ and if $\xi^\epsilon(0)$ follows the stationary distribution $\mathcal{N}(0, I_d)\otimes\mathcal{N}(0, I_d)$, the Markov process $(y^\epsilon_\former(t),v^\epsilon_\former(t),\alpha^\epsilon(t),\beta^\epsilon(t))$  converges to an ordinary differential equation with jumps which will be defined in (\suppEqBPSODE) in the supplementary material \citep{BierkensKamatanRoberts2024} in the Skorokhod topology. 
\end{theorem}

Let $N_T^\epsilon$ be the number of jumps of the process $(y^\epsilon(t),v^\epsilon(t))$ within the time interval of $[0,T]$. Based on the following theorem, we can conclude that $\mathbb{E}[N_T^\epsilon]=O(\epsilon^{-1} T)$. This implies that $O(\epsilon^{-1})$ jumps are required to achieve mixing. 
The proof of the following theorem is in Section \suppB.

\begin{theorem}
\label{theo:bps-costs}
    Under the same conditions as those for Theorem \ref{theo:bps-main}, the expected number of jumps of $(y^\epsilon(t),v^\epsilon(t))$ in the time interval $[0,T]$ is
\begin{align*}
\mathbb{E}[N_T^\epsilon]\le \frac{T}{2}(\replaced{\mathrm{tr}}{~\mathrm{diag}}((\Lambda^\epsilon)^{-1})+2\rho). 
\end{align*}
Also, 
\begin{align*}
\lim_{\epsilon\rightarrow 0}\epsilon~\mathbb{E}[N_T^\epsilon]=\frac{T}{2\sqrt{\pi}}~\frac{\Gamma((l+1)/2)}{\Gamma(l/2)}. 
\end{align*}
\end{theorem}

\subsection{Proof strategy}
\label{subsec:proofstrategy}

\added{Our approach is to employ a multiscale analysis \cite{Pavliotis2008} to determine the limiting behaviours of both the Zig-Zag sampler and the Bouncy Particle Sampler. Multiscale analysis serves as a comprehensive framework to simplify problems characterised by multiple scales. This framework is divided into two schemes: averaging and homogenisation. Averaging is often described as a first-order expansion and is seen as a result of the law of large numbers. Conversely, homogenisation is referred to as a second-order expansion and can emerge as a result of the central limit theorem. 

More specifically, we examine the 
backward Kolmogorov equation for the Zig-Zag sampler and the Bouncy Particle Sampler. We introduce a perturbation of the corresponding extended generator through a scale factor \(\epsilon\). Our objective is to expand this equation and identify the limit. By identifying the limit, then we will show convergences of the corresponding processes using martingale theory.

The hightest order term in the generator expansions describes the dynamics of the fast components. Heuristically, these average out to a stationary distribution, conditional on the slow components that may be considered fixed. The asymptotic dynamics of the slow components may then be computed by averaging the lower order term in the generator with respect to the obtained stationary distribution of the fast components. This provides a heuristic description of the \emph{averaging approach}.


In some cases the speed of the slow components averages out to have zero in the chosen time scale. In such cases an extra step of the expansion might be necessary, requiring us to look at an even longer, diffusive time scale. This is the situation where the \emph{homogenisation approach} is useful.
}

\added{
The averaging approach has been employed in the Markov chain Monte Carlo literature by works such as \cite{Sherlock2015}, \cite{beskos2018} and \cite{Sherlock2017}. However the application of the homogenisation approach seems to represent a novel contribution to this field.
}
We define $u(x,v,t)=\mathbb{E}[\phi(\xi^\epsilon(t))|\xi^\epsilon(0)=(x,v)]$ for a real-valued function $\phi$ where $\xi^\epsilon(t)$ represents either the Zig-Zag process or the bouncy particle process. Under certain regularity conditions, the function $u$ satisfies the backward Kolmogorov equation which can be expressed as
$$
\frac{\mathrm{d} u}{\mathrm{d} t}=\mathcal{L}^\epsilon u, 
$$
where $\mathcal{L}^\epsilon$ is the extended generator. 
The scaling limit of the extended generator $\mathcal{L}^\epsilon$ can be analysed by using the expansion of this differential equation in powers of a small parameter $\epsilon$. 

To understand the scaling limit of the Zig-Zag sampler, we employ the \textit{homogenisation} approach. Specifically, we solve the differential equation
\begin{equation}
\nonumber
\frac{\mathrm{d} u}{\mathrm{d} t}=\added{\epsilon^{-1}}~\mathcal{L}^\epsilon u, 
\end{equation}
where the factor $\epsilon$ on the left-hand side represents time scaling. 
The extended generator $\mathcal{L}^\epsilon$ is then expanded into $\mathcal{L}^\epsilon =\epsilon^{-1}\mathcal{L}_0+\mathcal{L}_1+\added{O(\epsilon)}$ and the function $u$ is expanded into
$$
u=u_0+\epsilon u_1+\epsilon^2 u_2+\added{O(\epsilon^3)}. 
$$
This leads to a system of equations
\begin{equation}
\label{eq:zig-zag-bk}
\begin{cases}
    0&=\quad\mathcal{L}_0u_0\\
    0&=\quad\mathcal{L}_0u_1+\mathcal{L}_1u_0\\
    \frac{\mathrm{d} u_0}{\mathrm{d} t}&=\quad\mathcal{L}_0u_2+\mathcal{L}_1u_1
\end{cases}
\end{equation}
at orders \replaced{$O(\epsilon^{-2}), O(\epsilon^{-1})$ and $O(1)$}{$O(\epsilon^{-1}), O(1)$ and $O(\epsilon)$} of the backward Kolmogorov equation. 
The system can be used to derive a second-order differential equation for a function $u_0$, given by
$$
 \frac{\mathrm{d} u_0}{\mathrm{d} t}= f(u_0),
$$
 which characterise the limit of the process $\xi^\epsilon(t)$. The expression of $f$ will be described later. 

The Bouncy Particle Sampler can be analysed using the {\em averaging} approach of \cite{Pavliotis2008}. This technique involves a different expansion of the Kolmogorov backward equation leading to  a set of equations given by 
\begin{align*}
\begin{cases}
    0&=\quad\mathcal{L}_0u_0\\
    \frac{\mathrm{d} u_0}{\mathrm{d} t}&=\quad\mathcal{L}_0u_1+\mathcal{L}_1u_0. 
\end{cases}
\end{align*}
Solving this set of equations leads to a first-order differential equation that characterises the scaling limit of the Bouncy Particle Sampler.

\section{Properties of the leading-order generators}\label{sec:leading}

We will structure our convergence proofs based on the methodology presented in Section \ref{subsec:proofstrategy}. Essentially, each proof, whether for Zig-Zag or Bouncy Particle Sampler, is divided into two parts. In the first part, we illustrate the properties of the leading-order generator. In the second part, we prove weak convergence via homogenization for the Zig-Zag sampler and through averaging for the Bouncy Particle Sampler. While the proof for the second part is lengthy, it's straightforward. Therefore, we've included all proofs related to this second part in the supplementary material \citep{BierkensKamatanRoberts2024}. Meanwhile, in the following sections of the main manuscript, we will focus exclusively on the first part.

\subsection{Properties of the leading-order generator for the Zig-Zag sampler}

The reparametrised extended generator of the Zig-Zag sampler (\ref{eq:reparametrise})  has a formal expansion
$$
\mathcal{L}^\epsilon=\epsilon^{-1}\mathcal{L}_0+\mathcal{L}_1+\added{O(\epsilon)}, 
$$
where 
\begin{equation}
\label{eq:zig-zag-operator-slow}
\begin{split}
   \mathcal{L}_0&=(\Lambda_\latter^{-1}(Uv)_\latter)^{\top}\partial_{y_\latter}
    +\sum_{i=1}^d (v_i(U_{_\latter,\cdot}^{\top}\Lambda_{\latter}^{-1}y_\latter)_i)_+~(\mathcal{F}_i-\operatorname{id})\\
    \mathcal{L}_1&=(\Lambda_\former^{-1}(Uv)_\former)^{\top}\partial_{y_\former}
    +\sum_{i=1}^d (v_i(U_{\former,\cdot}^{\top}\Lambda_\former^{-1}y_\former)_i)~1(v_i(U_{_\latter,\cdot}^{\top}\Lambda_\latter^{-1}y_\latter)_i\ge 0)~(\mathcal{F}_i-\operatorname{id}). 
\end{split}
\end{equation}
\added{For the expression of $\mathcal{L}_1$, we applied the expression $(a+\epsilon^{-1}b)_+ - \epsilon^{-1}b_+ - a \cdot 1(b > 0) = 0$ under the condition  $|a| \le \epsilon^{-1}|b|$. For detailed information, please refer to the supplementary material  \citep{BierkensKamatanRoberts2024}.}

We consider the operator $\mathcal{L}_0$ \replaced{as acting on real-valued function of $(y_\latter,v)$}{as an operator of $(y_\latter,v)$-valued functions}, where $y_\former$ is fixed. We establish the ergodicity of $\mathcal{L}_0$ by following the methodology outlined in \cite{10.1214/18-AAP1453}. However, we need to make some modifications to the arguments in \cite{10.1214/18-AAP1453} since the velocity in $\mathcal{L}_0$ is not aligned with the coordinates.
 
If $E$ is a topological space, a Markov process is called a $T$-process if there exists a probability measure $a$ on $[0,\infty)$, a lower semicontinuous, non-trivial semi Markov kernel $T$ such that 
$P_a(x, A)\ge T(x, A)\ (x\in E, A\in\mathcal{E})$ where 
$P_a(x,A)=\int P_t(x,A)a(\dif t)$. Let 
\begin{equation*}
    \mu=\mathcal{N}_\latter(0, I_\latter)\otimes\mathcal{U}(\{-1,+1\}^d). 
\end{equation*}

\begin{lemma}
\label{prp:zz-L0-invariance}
Under Assumption \ref{as:zz-condition}, 
the Markov process corresponding to the extended generator $\mathcal{L}_0$ is a $T$-process, ergodic and the invariant measure is $\mu$. 
\end{lemma}

\begin{proof}
The core of the generator is $C_c^\infty(\mathbb{R}^l\times\mathbb{R}^d)$
as proven in Corollary 6.2 of \cite{10.1214/21-ECP430}. 
Consequently,  $\mu$ is the invariant measure of the semigroup by Proposition 4.9.2 of \cite{MR838085}.
The rest of the proof is essentially the same as that of Section 2.2 and Section 3 of \cite{10.1214/18-AAP1453}. 
 Section 2.2 shows that for a multivariate normal target distribution, any state $z$ is reachable from any other state $z'$ in the sense that there is a piecewise deterministic path of $\mathcal{L}_0$ from $z'$ to $z$ such that the corresponding jump intensity is strictly positive at each jump point.
In our case, $e_1,\ldots, e_d$ in Corollary 2 of the paper are replaced by $\Lambda_\latter^{-1}(Ue_i)_\latter\ (i=1,\ldots, d)$, where $e_i\in\mathbb{R}^d$ is a vector which is $1$ at the $i$-th component and all other components are $0$. 
Therefore, the number of velocity vectors is larger than the dimension of $y_\latter$, and the velocity vectors are angled.  Nonetheless, we show that still, the argument in Section 2.2 is applicable to this case. 

We apply Lemmas 1-3 of their paper to this case. 
By Assumption \ref{as:zz-condition}, Lemmas 1 and 3 of their paper can be directly applicable and so we focus on Lemma 2.  The lemma  states that for any $y, y'\in\mathbb{R}^d$, the state $(y',-v)$ is reachable from $(y,v)$. 
Observe that 
$$
\Lambda_\latter^{-1}(Ue_i)_\latter=\Lambda_\latter^{-1}U_{\latter\latter}(e_i)_\latter\ (i\in\latter). 
$$
By applying a linear transformation $x\mapsto U_{\latter\latter}^{-1}\Lambda_\latter x$, the argument in Lemma 2 \added{of} \cite{10.1214/18-AAP1453} can be directly applied to the irreducibility proof of $\mathcal{L}_0$ if we omit $\Lambda_\latter^{-1}(Ue_i)_\latter\ (i\in\former)$. The proof will be carried out if we can handle the omitted component. On the other hand, thanks to Assumption \ref{as:zz-condition}, the components in $\former$ are asymptotically flippable in the sense that any coordinate can be switched after some time. Therefore, the argument in Lemma 2 can be applied to the current case after switching all $\former$ components. From this, the assertion follows. 
\end{proof}

We prove exponential ergodicity  of $\mathcal{L}_0$ by showing exponential drift inequality.  A slight modification of the drift function in  Lemma 11 of \cite{10.1214/18-AAP1453} can be applied in this case. 
In the following, we apply the linear operators $\mathcal{L}_0$ and $\mathcal{L}_1$ to vector valued functions coordinatewise.  

\added{In the supplementary material, we show that the first order term \(u_1\) of the perturbation solution to the backward Kolmogorov equation is constituted by a linear combination of \(\chi_0(y_\latter, v)\). Here,  \(\chi_0(y_\latter, v)\) is itself the solution to the equation \(\mathcal{L}_0\chi_0 = v\). The next proposition show the existence of \(\chi_0\).  }

\begin{proposition}
\label{prp:zz-L0-expergodicity}
The extended generator $\mathcal{L}_0$ is exponentially ergodic as an operator of the variable $(y_\latter, v)$. In particular, there is a solution $\chi_0:\mathbb{R}^l\times\mathbb{R}^d\rightarrow \mathbb{R}^d$ of the Poisson equation  $\mathcal{L}_0\chi_0(y_\latter,v)=v$, and it is unique up to a constant. Moreover, for any $c>0$ there is a constant $c_0$ such that $|\chi_0|\le c_0\exp(c|y_\latter|^2)$. In particular, $\chi_0(y_\latter,v)$ has an arbitrary order of moments. 
\end{proposition}

\begin{proof}
As in Lemma 11 of \cite{10.1214/18-AAP1453}, the exponential drift inequality can be determined by the drift function
$$
V(y_\latter,v)=\exp\left(\alpha~|y_\latter|^2/2+\sum_{i=1}^d\kappa(v_i(U_{\latter, \cdot}^{\top}\Lambda_\latter^{-1}y_\latter)_i)\right), 
$$
where $0<\alpha<1$ and $\kappa(s)=\operatorname{sgn}(s)\log(1+|s|)$. 
By Theorem 3.2 of \cite{Glynn_1996}, since $v$ is bounded, there is a solution to the Poisson equation for $\mathcal{L}_0\chi_0=v$ and the solution satisfies $|\chi_0|\le c_0(V+1)$ for some constant $c_0>0$. Thus the first claim follows since $\kappa(s)\le |s|\le 1+|s|^2$. Also, since we can take $\alpha\le 1/m$ for any $m\in\mathbb{N}$, the solution has the $m$-th order of moment. 
\end{proof}

\begin{lemma}\label{lem:domain-of-chi}
The solution $\chi_0(y_\latter,v)$ of the Poisson equation $\mathcal{L}_0\chi_0=v$ is in the domain of the extended generator $\mathcal{L}^\epsilon$, as a function of $(y,v)$. 
\end{lemma}

\begin{proof}
The solution $\chi_0$ is included in the domain of the extended generator $\mathcal{L}_0$ by construction. The domain of $\mathcal{L}^\epsilon$ is characterised by conditions (i-iii) of Theorem 5.5 of  \cite{MR790622}. The map $t\mapsto \chi_0(y_\latter+t~\Lambda_\latter^{-1}(Uv)_\latter, v)$ is absolutely continuous, so it satisfies condition (i). Moreover,  $|\chi_0|\le c_0\exp(c|y_\latter|^2)$ for constants $c, c_0>0$, which means that the local integrability condition for $(\mathcal{F}_i-\mathrm{id})\chi_0$ in (iii) of Theorem 5.5 is also satisfied for $i=1,\ldots, d$. The generator $\mathcal{L}^\epsilon$ is no boundary, so it satisfies condition (ii). By Theorem 5.5, the function $\chi_0$ is in the domain of the extended generator $\mathcal{L}^\epsilon$. 
\end{proof}

In the \replaced{supplementary material}{next section}, 
\added{we identify the limit equation satisfied by $u_0$ which is expressed through expectations of $\chi_0$. For this identification, we use a linear form of $\chi_0$ denoted as $\chi:\mathbb{R}^l\times\mathbb{R}^d\rightarrow\mathbb{R}^k$: }
\begin{equation}\label{eq:zz-poisson-sol}
\chi(y_\latter, v)=-(\Lambda_\former^{-1}(U\chi_0)_\former). 
\end{equation}
\added{
This is the solution of the Poisson equation}
\begin{equation}\label{eq:zz-poisson-eq}
\mathcal{L}_0\chi(y_\latter,v)=-\mathcal{L}_1\psi(y,v)=-(\Lambda_\former^{-1}(Uv)_\former). 
\end{equation}
\added{
The limit equation of $u_0$ is a second order differential equation of $y_\former$ corresponding to the  Ornstein--Uhlenbeck process
}
\begin{equation}
\label{eq:zz-ou}
\dif X(t)=-\frac{1}{2}\Upsilon~X(t)\dif t+\Omega^{1/2}~\dif W(t), 
\end{equation}
where $W(t)$ is the $k$-dimensional standard Wiener process, 
and $\Omega$ and $\Upsilon$ are $k\times k$ matrices defined by
\begin{equation}
\label{eq:upsilon}
\Omega=-~\mu((\mathcal{L}_0\chi)~\chi^{\top}+\chi~(\mathcal{L}_0\chi)^{\top}),\quad 
       \Upsilon=-2\mu\left(\chi~(\mathcal{L}_0\chi)^{\top}\right)~\Longrightarrow~
       \Omega=\frac{\Upsilon+\Upsilon^{\top}}{2}. 
\end{equation}

\begin{lemma}
The $k\times k$-matrix $\Omega$ is a symmetric positive definite matrix. 
\end{lemma}

\begin{proof}
Observe that the matrix can be expressed as
\begin{align*}
    \Omega
    =-\mu\left(\left((\mathcal{L}_0+\mathcal{L}_0^*)\chi~\right)\chi^{\top}\right)
\end{align*}
using the adjoint operator $\mathcal{L}_0^*$. 
Here, the adjoint operator is 
$$
\mathcal{L}_0^*=-(\Lambda_\latter^{-1}(Uv)_\latter)^{\top}\partial_{y_\latter}
    +\sum_{i=1}^d (-v_i(U_{_\latter,\cdot}^{\top}\Lambda_{\latter}^{-1}y_\latter)_i)_+~(\mathcal{F}_i-\operatorname{id}). 
$$
By the expression of the adjoint operator, we have 
$$
\mathcal{L}_0+\mathcal{L}_0^*=\sum_{i=1}^d a_i~(\mathcal{F}_i-\operatorname{id}),\quad a_i(y_\latter)=|(U_{_\latter,\cdot}^{\top}\Lambda_{\latter}^{-1}y_\latter)_i|. 
$$
Since $\mathcal{F}_i\mathcal{F}_i=\mathrm{id}$, we have
$$
-2\sum_{v\in \{-1,+1\}^d}\left(\left(\mathcal{F}_i-\mathrm{id}\right)\chi\right)\chi^{\top}=\sum_{v\in\{-1,+1\}^d}((\mathcal{F}_i-\mathrm{id})\chi)((\mathcal{F}_i-\mathrm{id})\chi)^{\top}\ge 0. 
$$
Therefore, 
\begin{align*}
    \Omega&=-\sum_{i=1}^d\mu\left(a_i~\left(\left(\mathcal{F}_i-\mathrm{id}\right)\chi~\right)\chi^{\top}\right)\ge 0
\end{align*}
Thus $\Omega$ is positive semidefinite. We show that it is indeed positive definite. If it is not positive definite, then there is a non-zero vector $u\in\mathbb{R}^k$ such that 
$$
u^{\top}\Omega u=\sum_{i=1}^d\mu(a_i~|u^{\top}(\mathcal{F}_i-\mathrm{id})\chi|^2)/2=0. 
$$
On the other hand, $a_i(y_\latter)=|(U_{_\latter,\cdot}^{\top}\Lambda_{\latter}^{-1}y_\latter)_i|\neq 0$  almost surely since the equality occurs when $y$ is on the $d-1$ dimensional hyperplane, which is a null set for the Lebesgue measure. Therefore, the above equation implies that $u^{\top}(\mathcal{F}_i-\mathrm{id})\chi=(\mathcal{F}_i-\mathrm{id})(u^{\top}\chi)=0$ almost surely for $i=1,\ldots d$. This implies that $u^{\top}\chi$ is in the null set of $\mathcal{F}_i-\mathrm{id}$ for any $i$. Therefore,  $u^{\top}\chi$ does not depend on $v$. Thus
$$
-u^{\top}(\Lambda_\former^{-1}(Uv)_\former)=\mathcal{L}_0(u^{\top}\chi)=
u^{\top}\partial_{y_\latter}\chi(y_\latter)~(\Lambda_\latter^{-1}(Uv)_\latter)\quad v\in\{-1,+1\}^d, 
$$
where $\partial_{y_\latter}\chi=(\partial_{y_j}\chi_i)_{i\in\former, j\in\latter,}$ is a $k\times l $ matrix. 
Since the above equation is linear in $v$, the equation holds for any $v\in\mathbb{R}^d$. However, the left-hand side depends only on $(Uv)_\former$ and the right-hand side depends only on $(Uv)_\latter$, so the equality holds only if both ends are $0$. However, this is impossible since the left side is strictly positive if $u=(Uv)_\former\neq 0$. Thus $\Omega$ is positive definite. 
\end{proof}

\begin{lemma}
\label{lem:zz-identity}
$\mu((\mathcal{L}_1(\chi\circ\psi))(y_\former,\cdot) )=-\Upsilon~y_\former/2$, and $\mu(m)=\Omega$ for 
\begin{align}
\label{eq:zz-diffusion-integrand}
    m(y_\latter, v)=\sum_{i=1}^d(v_i(U_{\latter\latter}^{\top}\Lambda_\latter^{-1}y_\latter)_i)_+\left((\mathcal{F}_i-\mathrm{id})\chi(y_\latter, v)\right)^{\otimes 2}. 
\end{align}
\end{lemma}

\begin{proof}
By construction, we have
\begin{align*}
\begin{cases}
    a_i&:=v_i(U_{\former,\cdot}^{\top}\Lambda_\former^{-1}y_\former)_i,\\
    b_i&:=v_i( U_{\latter,\cdot}^{\top}\Lambda_\latter^{-1}y_\latter)_i
\end{cases}\quad
\Longrightarrow
\quad
    \mathcal{L}_1(\chi\circ\psi)(y,v)=\sum_{i=1}^d a_i~1(b_i\ge 0)(\mathcal{F}_i-\mathrm{id})\chi(y_\latter, v). 
\end{align*}
Observe that 
\begin{equation}
\label{eq:zz-identity}
\mathcal{F}_ia_i=-a_i,\  
\mathcal{F}_ib_i=-b_i\quad\Longrightarrow\quad
    (\mathcal{F}_i-\mathrm{id})(a_i1(b_i\ge0))=-a_i-a_i1(b_i= 0). 
\end{equation}
From this fact, since $(\mathcal{F}_i-\operatorname{id})$ is $\mu$-reversible, $\mu((\mathcal{L}_1(\chi\circ\psi))(y_\former,\cdot) )$ is
\begin{align*}
\sum_{i=1}^d\mu\left( a_i~1(b_i\ge 0)~(\mathcal{F}_i-\operatorname{id})\chi\right)&=\sum_{i=1}^d
\mu\left(\chi~(\mathcal{F}_i-\operatorname{id}) a_i1(b_i\ge 0)\right)\\
&=
-\sum_{i=1}^d
\mu\left(\chi~a_i\right)
=-\frac{\Upsilon}{2}~ y_\former.
\end{align*}
Thus the first claim follows.

For the second claim, let
\begin{align*}
\begin{cases}
        \alpha_i&:=v_i(U_{\latter,\cdot}^{\top}\Lambda^{-1}_{\latter}y_{\latter})_i,\\
    \beta&:= (Uv)_\latter
\end{cases}\quad
\Longrightarrow
\quad
    \mathcal{L}_0=\beta^{\top}\partial_{y_\latter}+\sum_{i=1}^d(\alpha_i)_+(\mathcal{F}_i-\operatorname{id}). 
\end{align*}
Let $\chi_i$ be the $i$-th component of $\chi(y_\latter,v)\in\mathbb{R}^d$. Observe that there are two identities
\begin{alignat*}{4}
    &(\mathcal{F}_i-\operatorname{id})(\chi_k)\chi_l
    &&+(\mathcal{F}_i-\operatorname{id})(\chi_l)\chi_k
    &&-(\mathcal{F}_i-\operatorname{id})(\chi_k\chi_l)
    &&=-(\mathcal{F}_i-\operatorname{id})(\chi_k)(\mathcal{F}_i-\operatorname{id})(\chi_l),\\
    &(\partial_{y_\latter}\chi_k)\chi_l
    &&+(\partial_{y_\latter}\chi_l)\chi_k
    &&-\partial_{y_\latter}(\chi_k\chi_l)
    &&=0, 
\end{alignat*}
where we used the fact that $\mathcal{F}_i(fg)=\mathcal{F}_i(f)\mathcal{F}_i(g)$.
From this fact,  the $(k,l)$-th component of $m(y_\latter, v)$ is, 
\begin{align*}
    m_{kl}(y_\latter, v)&=\sum_{i=1}^d(\alpha_i)_+(\mathcal{F}_i-\operatorname{id})(\chi_k)(\mathcal{F}_i-\operatorname{id})(\chi_l)\\
    &=-\mathcal{L}_0(\chi_l)\chi_k-\mathcal{L}_0(\chi_k)\chi_l+\mathcal{L}_0(\chi_l\chi_k). 
\end{align*}
 Since $\mathcal{L}_0$ is $\mu$-invariant, the expectation of the third term in the right-hand side is $0$. Thus we have
\begin{align*}
   \mu(m)&~=
    -\int (\mathcal{L}_0(\chi)\chi^{\top}+\chi\mathcal{L}_0(\chi)^{\top})\mu(\dif y_\latter,\dif v)=\Omega. 
\end{align*}
\end{proof}




\subsection{Properties of the leading-order generator of the Bouncy Particle Sampler}

The generator of the Bouncy Particle Sampler (\ref{eq:bps-operator}) has a formal expansion 
\begin{equation}
\label{eq:bps-generator-expansion}
\mathcal{L}^\epsilon=\epsilon^{-1}\mathcal{L}_0+\mathcal{L}_1+o(\epsilon), 
\end{equation}
where the leading term is 
\begin{equation}
\label{eq:bps-operator-fast}
    \mathcal{L}_0=(\Lambda_\latter^{-1}v_\latter)^{\top}\partial_{y_\latter}+(v_\latter^{\top}\Lambda_\latter^{-1}y_\latter)_+(\mathcal{B}^0-\operatorname{id})
\end{equation}
with a reflection operator $\mathcal{B}^0$  defined by  $\mathcal{B}^0 f(x,v)=f(x,B^0(y)v)$ where
$$
B^0(y)v
=
\begin{pmatrix}
v_\former\\
v_\latter-2 v_\latter^{\top}\Lambda_\latter^{-1}y_\latter\frac{\Lambda_\latter^{-1}y_\latter}{|\Lambda_\latter^{-1}y_\latter|^2}
\end{pmatrix}. 
$$
The operator $\mathcal{L}_0$ serves as an extended generator for the Markov process associated with the Bouncy Particle Sampler, in the absence of refreshment jumps (as outlined in Theorem 5.5 of \cite{MR790622}). The operator $\mathcal{L}_1$ will be discussed subsequently. In order to conduct multivariate analysis, it is necessary to identify the null space of $\mathcal{L}_0$. The null space is contingent upon $\Lambda_\latter$ and necessitates a distinct analysis for varying null spaces. To simplify the analysis, we shall only consider the simplest case.

\begin{assumption}
\label{asp:bps}
$\Lambda_\latter= I_\latter$. 
\end{assumption}

\noindent
In this scenario, the Markov process $(y_\latter(t), v_\latter(t))$ with extended generator $\mathcal{L}_0$ is by design constrained to the two-dimensional hyperplane spanned by $y_\latter(0)$ and $v_\latter(0)$, as depicted in Figure \ref{fig:bouncy_no_refresh}. By specifying two linearly independent vectors as the basis for this hyperplane, the dynamics of $y_\latter(t)$ and $v_\latter(t)$ are determined by at most four scalar parameters.
Moreover, 
\begin{equation}
\label{eq:bps-constant}
\begin{cases}
\alpha&:=|v_\latter|^2\\
\beta&:=|v_\latter|^2|y_\latter|^2
-(v_\latter^{\top}y_\latter)^2
\end{cases}
\end{equation}
and 
\begin{equation*}
    r:=\sqrt{\beta/\alpha}
\end{equation*}
do not change through the dynamics. 
Let $f:\mathbb{R}\times\mathbb{S}\rightarrow\mathbb{R}$ such that $f(z,e)$ is differentiable with respect to $z$. Consider the operator
\begin{equation}
\label{eq:H}
Hf(z,e)=\alpha^{1/2}\left[\partial_zf(z,e)+z_+(f(-z,b(z,e))-f(z,e))\right]. 
\end{equation}
In (\ref{eq:e_jump}), we will provide the expression for the vector $b(z,e)$, which can be uniquely determined based on the fixed orientation and the variables $z$ and $e$.
Let $(z(t), e(t))$ be the transformation of the process $(y_\latter(t),v_\latter(t))$ by 
\begin{equation}
\label{eq:bps-ergodic-component}
z=|v_\latter|^{-1}v_\latter^{\top}y_\latter,\ \quad
e= r^{-1}\left(y_\latter-zv_\latter/|v_\latter|\right).
\end{equation}
Observe that $e$ is a unit vector.
Note that if $z(t)$ is continuous at $t$, then $z(t)'=|v_\latter(t)|$, and if not, $z(t)=-z(t-)$. We show that the process has the extended generator $H$. 

\begin{figure}[h]
\centering
\includegraphics[width=0.75\linewidth]{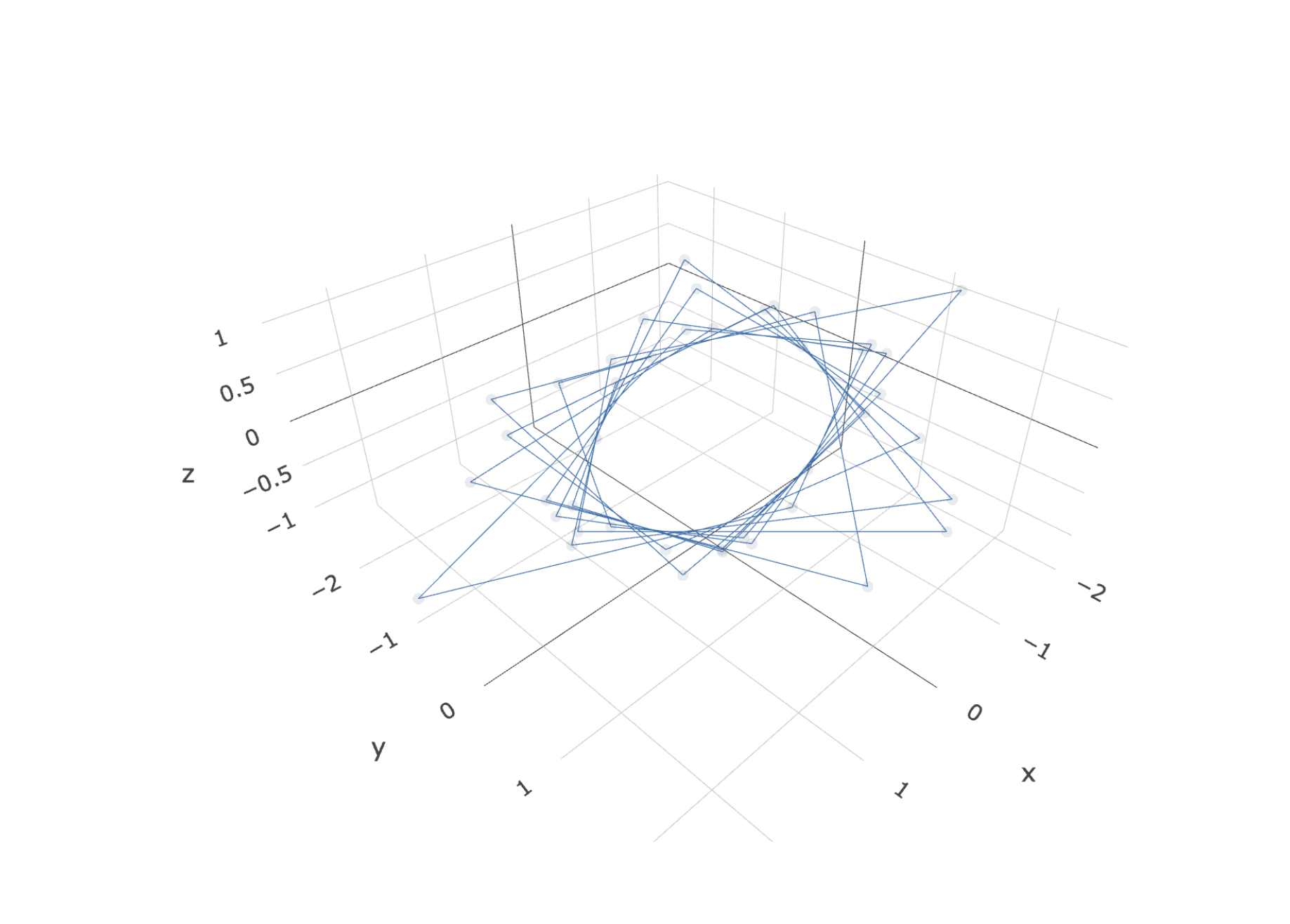}
\caption{Typical path of the process  $y_\latter(t)$. }
\label{fig:bouncy_no_refresh}
\end{figure}
 
\begin{figure}[h]
\centering
\includegraphics[width=0.5\linewidth]{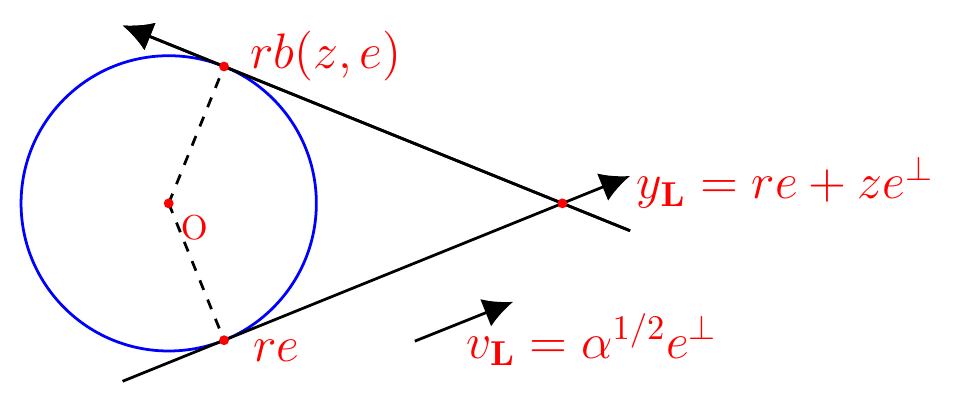}
\caption{Trajectory of  $y_\latter(t)$ }
\label{fig:bps_degenerate_path}
\end{figure}

\begin{proposition}\label{prp:bps-leading-order-ergodicity}
\replaced{Under}{Assuming that} Assumption \ref{asp:bps}, \deleted{holds, and} given the vectors $y_\latter(0)$ and $v_\latter(0)$, which are assumed to be linearly independent, we define $\Pi$ as the oriented hyperplane spanned by these vectors. Then, we can assert that $(z(t), e(t))$ forms an ergodic Markov process with extended generator $H$, and its invariant distribution is $\mathcal{N}(0,1)\otimes \mathcal{U}(\mathbb{S})$.
\end{proposition}

\begin{proof}
Observe that linearly independence implies $\alpha> 0$. 
Let $0=\tau_0<\tau_1<\cdots$ be the jump times of the Markov process $(y_\latter(t), v_\latter(t))$. For each $n=0,1,\ldots$, 
$v_\latter(t)=v_\latter(\tau_n)$ and 
$$
y_\latter(t)=
y_\latter(\tau_n)+v_\latter(\tau_n)~(t-\tau_n)
$$ forms a line segment in each time interval $t\in [\tau_n,  \tau_{n+1})$. By the constraint equation  (\ref{eq:bps-constant}),
\begin{equation}\nonumber
r^2=|y_\latter|^2-(v_\latter^{\top}y_\latter)^2/|v_\latter|^2 \quad\Longrightarrow\quad |y_\latter|\ge r
\end{equation}
and equality holds if and only if $z=y_\latter^{\top} v_\latter=0$. Therefore, the line segment is tangent to the circle of radius $r$. Also, since 
$$
z(t)=z(\tau_n)+(t-\tau_n)\alpha^{1/2},\quad t\in [\tau_n,\tau_{n+1})
$$
and $z(\tau_n)=-z(\tau_n-)<0$, 
the tangency point $z(t)=0$ is included in each line segment. From this observation, for each line segment, we can define the orientation, clockwise or counterclockwise, by observing the direction at the tangency point. 

We claim that the orientation of all line segments are consistent. Otherwise, there exists a set of line segments with differing orientations, such as $(y_\latter(t):\tau_n\le t<\tau_{n+1})$ and $(y_\latter(t):\tau_{n+1}\le t<\tau_{n+2})$. At $y(\tau_{n+1})$, there are two tangent lines to the two-dimensional circle of radius $r$. If the two line segments are on different tangent lines, then the orientation must be the same, given that the path $(y_\latter(t):\tau_n\le t<\tau_{n+2})$ is connected. Thus, if orientations are different, the two line segments are on the same tangent line, and the orientations of the two line segments must have opposite signs, i.e. $v_\latter(\tau_{n+1}-)=-v_\latter(\tau_{n+1})$. This can only occur when $v_\latter(t)/|v_\latter(t)|=y_\latter(t)/|y_\latter(t)|$ at $t=\tau_{n+1}-$. However, by turning back the time it implies linearly dependence of  $v_\latter(0)$ and $y_\latter(0)$ which contradicts the initial assumption. 

Observe that
\begin{equation}
\nonumber
    e^{\perp}:=\alpha^{-1/2}~v_\latter\quad\Longrightarrow\quad z=y_\latter^{\top}e^{\perp}\ \mathrm{and}\ e=r^{-1}(y_\latter-z~e^{\perp})\quad\Longrightarrow\quad e\perp e^{\perp}. 
\end{equation}
The unit vector $e^\perp$ is solely determined by the vector $e$ together with the orientation of the segments. 
By construction, $e\perp v_\latter$ and  $r~e=y_\latter-zv_\latter/|v_\latter|$ is on the line segment. Thus $r~e$ is the point of tangency. See Figure \ref{fig:bps_degenerate_path}. 
For given $\Pi$, $\alpha, \beta$ and the orientation of the process,  there is a map $\varphi(z, e)=(y_\latter, v_\latter)$ from $\mathbb{R}\times \mathbb{S}$ to $\mathbb{R}^l\times\mathbb{R}^l$ defined by 
\begin{equation*}
v_\latter = \alpha^{1/2}e^{\perp},\quad\ y_\latter = re+ze^{\perp}. 
\end{equation*}
If there is a reflection jump event at $t>0$, then  $e^{\perp}$ will jump to 
\begin{align*}
(e^{\perp})^*=e^{\perp}-2\frac{v_\latter^{\top}y_\latter}{|v_\latter|}\frac{y_\latter}{|y_\latter|^2}
=\frac{-2rz}{r^2+z^2}e+\frac{r^2-z^2}{r^2+z^2}e^{\perp}, 
\end{align*}
where we used $|y_\latter|^2=r^2+z^2$. 
Also, $z$ jumps to $-z$, and $e$ jumps to 
\begin{equation}
\label{eq:e_jump}
b(z,e)=r^{-1}(y_\latter-(-z)~(e^{\perp})^*)=\frac{r^2-z^2}{r^2+z^2}e+\frac{2rz}{r^2+z^2}e^{\perp}.     
\end{equation}

 Since $(y_\latter(t), v_\latter(t))$ is a Markov process, $(z(t), e(t))$ is also a Markov process. The Markov process $(z(t), e(t))$ is ergodic if there is an invariant probability measure by Proposition \ref{prop:irreducibility} below and Theorem 1 of \cite{Kulik_2015}. We show that the 
 invariant probability measure is $\mathcal{N}(0,1)\otimes\mathcal{U}(\mathbb{S})$. To see this, observe that the extended generator of the Markov process $(z(t), e(t))$ is 
 $$
 \widetilde{H}f(z,e)
     =\mathcal{L}_0(f\circ\varphi^{-1})(\varphi(z,e))
 $$
 since for $Y_t:=(z(t),e(t))$, $X_t=\varphi(X_t)=(y_\latter(t),v_\latter(t))$ and  $g=f\circ\varphi^{-1}$, the process
 \begin{align*}
     M_T:=f(Y_T)-f(Y_0)-\int_0^T(\widetilde{H}f)(Y_t)\dif t&=
     g(X_T)-g(X_0)-\int_0^T(\mathcal{L}_0g)(X_t)\dif t
 \end{align*}
 is a local martingale when $g$ is the domain of the extended generator of $\mathcal{L}_0$. 
 We see that the extended generator $\widetilde{H}$ coincides $H$ in (\ref{eq:H}). 
 Observe that we have the expression
 $$
 (z,e)=\varphi^{-1}(y_\latter,v_\latter)=(\alpha^{-1/2}~y_\latter^{\top}v_\latter, r^{-1}(y_\latter -(\alpha^{-1/2}y_\latter^{\top}v_\latter)~(\alpha^{-1/2}v_\latter))).$$
 By differentiating both sides of the equation, we obtain 
 \begin{align*}
     v_\latter^{\top}\partial_{y_\latter}z=\alpha^{1/2},\ \quad
     v_\latter^{\top}\partial_{y_\latter}e=0. 
 \end{align*}
 Together with the fact that the $(z,e)$ jumps to $(-z, b(z,e))$ at the reflection jump, the expression (\ref{eq:H}) follows. 
 By integration by parts formula, we have
 $$
 \int_{-\infty}^\infty \left\{\partial_z f(z,e)+z_+(f(-z,e)-f(z,e))\right\}\phi(z)\dif z=\int_{-\infty}^\infty \left\{\partial_z f(z,e)-zf(z,e)\right\}\phi(z)\dif z=0
 $$
 for any absolutely continuous function $f(\cdot,e)$ such that $\int|\partial_zf(z,e)|\phi(z)\dif z<\infty$ for each $e\in\mathbb{S}$. 
 On the other hand, 
 $$
 \int_{\mathbb{S}} f(-z, b(z,e))\dif e=
 \int_{\mathbb{S}} f(-z, e)\dif e. 
 $$
 Therefore, the expectation of $Hf$ with respect to $ \mathcal{N}(0,1)\otimes\mathcal{U}(\mathbb{S})$ is always $0$. Hence the claim follows from Theorem 21 and Corollary 22 of \cite{Durmus_2021}. 
\end{proof}

\begin{proposition}\label{prop:irreducibility}
Let $P_t$ be the Markov kernel of $(z(t), e(t))$.
Then for $(z,e), (z^*, e^*)\in\mathbb{R}\times\mathbb{S}$, there exists $T>0$ such that 
\begin{equation}\label{eq:irreducibility}
P_t((z,e), \cdot)\not\perp
P_t((z^*,e^*), \cdot)
\end{equation}
for all $t\ge T$. 
\end{proposition}

\added{The proof of the proposition is provided in the supplementary material \cite{BierkensKamatanRoberts2024}.}

\begin{corollary}
\label{cor:bps_constant}
Let $y_\latter(0), v_\latter(0)\in\mathbb{R}^\latter$ be fixed constants such that the two vectors are linearly independent. Then 
\begin{align*}
    \mathbb{E}\left[\frac{(v_\latter(t)^{\top}y_\latter(t))_+^2}{|y_\latter(t)|^2}\right]
    =
    \mathbb{E}\left[\frac{\alpha~ z(t)_+^2}{\beta/\alpha+z(t)^2}\right]\longrightarrow_{t\rightarrow\infty}
    c(\alpha,\beta)
\end{align*}
where, for $r=(\beta/\alpha)^{1/2}$,  
\begin{equation}
\label{eq:bps-averaged-quantity}
c(\alpha, \beta)
=\alpha-\sqrt{2\pi}~\alpha re^{r^2/2}\Phi(-r)=\alpha\left[1-r\frac{\Phi(-r)}{\phi(r)}\right], 
\end{equation}
where $\Phi$ is the cumulative distribution function of the normal distribution. 
\end{corollary}

\begin{proof}
The equation (\ref{eq:bps-constant}) yields the first equation. Since the invariant measure of $z(t)$ is the standard normal distribution, by the law of large numbers and by Proposition \ref{prp:bps-leading-order-ergodicity}, the expectation is
\begin{align*}
    \int_\mathbb{R}\frac{\alpha x^2}{r^2+x^2}\phi(x)\dif x
    =\alpha
    -
    \alpha r^2~\int_\mathbb{R}\frac{1}{r^2+x^2}\phi(x)\dif x. 
\end{align*}
Since the characteristic function of the probability density function $\phi(x)$ of the standard normal distribution is $\sqrt{2\pi}\phi(x)$, we have
\begin{align*}
    \int_\mathbb{R}\frac{1}{r^2+x^2}\phi(x)\dif x&=
    (2\pi)^{-1/2}\int_\mathbb{R}\int_\mathbb{R}\frac{1}{r^2+x^2}e^{\mathbf{i}xu}\phi(u)\dif x\dif u\\
    &=
    (2\pi)^{-1/2}r^{-1}\int_\mathbb{R}\int_\mathbb{R}\frac{1}{1+x^2}e^{\mathbf{i}xur}\phi(u)\dif x\dif u\\
    &=
    (\pi/2)^{1/2}r^{-1}\int_\mathbb{R} e^{-|u|r}\phi(u)\dif u
\end{align*}
where we used the fact that the characteristic function of the Cauchy distribution $\pi^{-1}(1+x^2)^{-1}\dif x$ is $\exp(-|x|)$. 
This expectation is
\begin{align*}
    2~(\pi/2)^{1/2}r^{-1}\int_0^\infty e^{-u r}\phi(u)\dif u
    =2~(\pi/2)^{1/2}r^{-1}~\Phi(-r)e^{r^2/2}. 
\end{align*}
This proves the claim. 
\end{proof}


\section{Discussion}
\label{sect:disc}

In this study, we aimed to determine the scaling limit of the Zig-Zag sampler and the Bouncy Particle Sampler for a special class of anisotropic target distributions. 
Here, we would like to compare PDMPs with MCMC algorithms, such as the random walk Metropolis algorithm. As described in \cite{beskos2018}, the computational effort for the random walk Metropolis algorithm is $O(\epsilon^{-1})$ \added{when the dimension of the smaller component $l$ equals $1$}, based on the Markov jump process limit. However, in general, the computational complexity is $O(\epsilon^{-2})$ as discussed in \cite{beskos2018}. Therefore, even after selecting the optimal scaling for the random walk Metropolis algorithm, the BPS has a better complexity of $O(\epsilon^{-1})$, and the ZZS has the same complexity of $O(\epsilon^{-2})$ when $l\neq 1$, in terms of orders of magnitude of $\epsilon^{-1}$. We estimate that the cost per unit time for PDMPs is on the order of $\epsilon^{-1}$ by Theorems \ref{theo:zz-costs} and \ref{theo:bps-costs}.
Thus, our theoretical results demonstrate that the Bouncy Particle Sampler exhibits superior convergence rate compared to random walk Metropolis chains for anisotropic target distributions, supporting the use of piecewise deterministic Markov processes for anisotropic target distributions.

As we discussed in  Subsection \ref{sec:2dim-discussion},
it is important to emphasise that our complexity analysis is strongly related to the implementational cost of either of these algorithms, and this cost can vary substantially between different applications.
Also target distributions with significant conditional independence structure can be less expensive to explore using Zig-Zag than the Bouncy Particle Sampler, see for example
\cite{BierkensFearnheadRoberts2016,chevallier2023reversible}, and this could offset the theoretical advantages that BPS has in these contexts.

Our findings also highlight the importance of  pre-conditioning of the Zig-Zag sampler, which is a coordinate-dependent method. One approach to improving the performance of the Zig-Zag sampler is to transpose the state space so that the components are roughly uncorrelated. It is helpful to set the scale and direction such that the number of jumps in each direction are roughly equal. However, the effective direction and scale of the state space can be position-dependent. In these cases, the Bouncy Particle Sampler may have an advantage due to its relative simplicity.

The asymptotic behaviour of these processes is sensitive to the target distribution. In particular, we find that the Bouncy Particle Sampler can become degenerate and trapped in a low-dimensional subspace in the limit if there are no refreshment jumps. The specific low-dimensional subspace is determined by the structure of the target distribution. 
Also, we only studied the simplest scenario for the covariance matrix. More general scenario such as more than two scales were out of scope for this paper. 

Multi-scale analysis provides an efficient framework for the scaling limit analysis of both Markov chains and Markov processes. By using this approach, we are able to effectively separate the contributions of the two scales of the Markov processes using the solution of the Poisson equation. This also simplifies the structure of the proof. We believe that this framework could be useful for the scaling limit analysis of many other Markov chains and processes.

\added{
Finally, recent work  by \cite{andrieu2021hypocoercivity} provides $L^2$-exponential convergence for general target distributions, identifying the exponential rate even in high-dimensional scenarios. Their results can also be applied to anisotropic target distributions and provide an upper bound for the convergence rates. However, it is worth noting that it is currently unclear how to achieve the convergence rate described in this paper as an upper bound for their study.
}


\begin{funding}
JB was funded by the research programme ‘Zigzagging through computational barriers’ with project number 016.Vidi.189.043, which was financed by the Dutch Research Council (NWO). KK was supported in part by JST, CREST Grant Number JPMJCR2115, Japan, and JSPS KAKENHI Grant Numbers JP20H04149 and JP21K18589.
GOR was supported by EPSRC grants Bayes for Health (R018561) CoSInES (R034710) PINCODE (EP/X028119/1), EP/V009478/1 and also
UKRI for grant EP/Y014650/1, as part of the ERC Synergy project OCEAN.
\end{funding}

\begin{supplement}
\stitle{Supplement to ``Scaling of piecewise deterministic Monte Carlo for anisotropic targets''}
\sdescription{The supplementary material provides proofs for the results presented in Sections \ref{subsec:zigzag}, \ref{subsec:bps}, and for Proposition 4.8. Additionally, it includes the proof of the results for the two-dimensional Zig-Zag sampler analysis described in Section 2.1.}
\end{supplement}


\bibliographystyle{imsart-nameyear} 

\begin{thebibliography}{30}

\bibitem[Andrieu and Thoms(2008)]{andrieu2008tutorial}
Andrieu, C. and Thoms, J. (2008). A tutorial on adaptive MCMC. \textit{Stat. Comput.} \textbf{18} 343--373.

\bibitem[Andrieu et al.(2021)]{andrieu2021hypocoercivity}
Andrieu, C., Durmus, A., Nüsken, N. and Roussel, J. (2021). Hypocoercivity of piecewise deterministic Markov process-Monte Carlo. \textit{Ann. Appl. Probab.} \textbf{31} 2478--2517.

\bibitem[Beskos et al.(2013)]{beskos2013optimal}
Beskos, A., Pillai, N., Roberts, G., Sanz-Serna, J.-M. and Stuart, A. (2013). Optimal tuning of the hybrid Monte Carlo algorithm. \textit{Bernoulli} \textbf{19} 1501--1534.

\bibitem[Beskos et al.(2018)]{beskos2018}
Beskos, A., Roberts, G., Thiery, A. and Pillai, N. (2018). Asymptotic analysis of the random walk Metropolis algorithm on ridged densities. \textit{Ann. Appl. Probab.} \textbf{28} 2966--3001. https://doi.org/10.1214/18-AAP1380

\bibitem[Bierkens and Duncan(2017)]{BierkensDuncan2016}
Bierkens, J. and Duncan, A. (2017). Limit theorems for the zig-zag process. \textit{Adv. Appl. Probab.} \textbf{49} 791--825. https://doi.org/10.1017/apr.2017.22

\bibitem[Bierkens, Fearnhead and Roberts(2019)]{BierkensFearnheadRoberts2016}
Bierkens, J., Fearnhead, P. and Roberts, G. (2019). The zig-zag process and super-efficient sampling for Bayesian analysis of big data. \textit{Ann. Statist.} \textbf{47} 1288--1320. https://doi.org/10.1214/18-aos1715

\bibitem[Bierkens, Kamatani and Roberts(2022)]{bierkens2018high}
Bierkens, J., Kamatani, K. and Roberts, G. O. (2022). High-dimensional scaling limits of piecewise deterministic sampling algorithms. \textit{Ann. Appl. Probab.} \textbf{32} 3361--3407.

\bibitem[Bierkens, Kamatani and Roberts(2024)]{BierkensKamatanRoberts2024}
Bierkens, J., Kamatani, K. and Roberts, G. O. (2024). Supplement to `Scaling of piecewise deterministic Monte Carlo for anisotropic targets''.


\bibitem[Bierkens, Roberts and Zitt(2019)]{10.1214/18-AAP1453}
Bierkens, J., Roberts, G. O. and Zitt, P.-A. (2019). Ergodicity of the zigzag process. \textit{Ann. Appl. Probab.} \textbf{29} 2266--2301. https://doi.org/10.1214/18-AAP1453

\bibitem[Billingsley(1999)]{MR1700749}
Billingsley, P. (1999). \textit{Convergence of Probability Measures}, 2nd ed. Wiley Series in Probability and Statistics: Probability and Statistics. John Wiley \& Sons Inc., New York. https://doi.org/10.1002/9780470316962

\bibitem[Bouchard-Côté, Vollmer and Doucet(2018)]{BouchardCote2017}
Bouchard-Côté, A., Vollmer, S. J. and Doucet, A. (2018). The bouncy particle sampler: A nonreversible rejection-free Markov chain Monte Carlo method. \textit{J. Amer. Statist. Assoc.} \textbf{113} 855--867. https://doi.org/10.1080/01621459.2017.1294075

\bibitem[Chevallier, Fearnhead and Sutton(2023)]{chevallier2023reversible}
Chevallier, A., Fearnhead, P. and Sutton, M. (2023). Reversible jump PDMP samplers for variable selection. \textit{J. Amer. Statist. Assoc.} \textbf{118} 2915--2927.

\bibitem[Davis(1984)]{MR790622}
Davis, M. H. A. (1984). Piecewise-deterministic Markov processes: A general class of nondiffusion stochastic models. \textit{J. Roy. Statist. Soc. Ser. B} \textbf{46} 353--388. https://doi.org/10.1111/j.2517-6161.1984.tb01287.x

\bibitem[Deligiannidis, Paulin and Doucet(2021)]{deligiannidisrandomized}
Deligiannidis, G., Paulin, D. and Doucet, A. (2021). Randomized Hamiltonian Monte Carlo as scaling limit of the bouncy particle sampler and dimension-free convergence rates. \textit{Ann. Appl. Probab.} \textbf{31} 2612--2662.

\bibitem[Durmus, Guillin and Monmarché(2021)]{Durmus_2021}
Durmus, A., Guillin, A. and Monmarché, P. (2021). Piecewise deterministic Markov processes and their invariant measures. \textit{Ann. Inst. Henri Poincaré Probab. Stat.} \textbf{57}. https://doi.org/10.1214/20-aihp1125

\bibitem[Ethier and Kurtz(1986)]{MR838085}
Ethier, S. N. and Kurtz, T. G. (1986). \textit{Markov Processes: Characterization and Convergence}. Wiley Series in Probability and Mathematical Statistics. John Wiley \& Sons Inc., New York. https://doi.org/10.1002/9780470316658

\bibitem[Glynn and Meyn(1996)]{Glynn_1996}
Glynn, P. W. and Meyn, S. P. (1996). A Liapounov bound for solutions of the Poisson equation. \textit{Ann. Probab.} \textbf{24}. https://doi.org/10.1214/aop/1039639370

\bibitem[Graham, Thiery and Beskos(2022)]{Graham2022}
Graham, M. M., Thiery, A. H. and Beskos, A. (2022). Manifold Markov chain Monte Carlo methods for Bayesian inference in diffusion models. \textit{J. Roy. Statist. Soc. Ser. B} \textbf{84}. https://doi.org/10.1111/rssb.12497

\bibitem[Holderrieth(2021)]{10.1214/21-ECP430}
Holderrieth, P. (2021). Cores for piecewise-deterministic Markov processes used in Markov chain Monte Carlo. \textit{Electron. Commun. Probab.} \textbf{26} 1--12. https://doi.org/10.1214/21-ECP430

\bibitem[Kulik and Scheutzow(2015)]{Kulik_2015}
Kulik, A. and Scheutzow, M. (2015). A coupling approach to Doob's theorem. \textit{Rendiconti Lincei - Mat. Appl.} \textbf{26} 83--92.

\bibitem[Pavliotis and Stuart(2008)]{Pavliotis2008}
Pavliotis, G. A. and Stuart, A. M. (2008). \textit{Multiscale Methods: Averaging and Homogenization}. Texts in Applied Mathematics. Springer, New York.

\bibitem[Peters and de With(2012)]{PhysRevE.85.026703}
Peters, E. A. J. F. and de With, G. (2012). Rejection-free Monte Carlo sampling for general potentials. \textit{Phys. Rev. E} \textbf{85} 026703. https://doi.org/10.1103/PhysRevE.85.026703

\bibitem[Roberts, Gelman and Gilks(1997)]{RGG}
Roberts, G. O., Gelman, A. and Gilks, W. R. (1997). Weak convergence and optimal scaling of random walk Metropolis algorithms. \textit{Ann. Appl. Probab.} \textbf{7} 110--120.

\bibitem[Roberts and Rosenthal(1998)]{roberts1998optimal}
Roberts, G. O. and Rosenthal, J. S. (1998). Optimal scaling of discrete approximations to Langevin diffusions. \textit{J. Roy. Statist. Soc. Ser. B} \textbf{60} 255--268.

\bibitem[Roberts and Rosenthal(2001)]{roberts2001optimal}
Roberts, G. O. and Rosenthal, J. S. (2001). Optimal scaling for various Metropolis-Hastings algorithms. \textit{Statist. Sci.} \textbf{16} 351--367.

\bibitem[Roberts and Rosenthal(2009)]{roberts2009examples}
Roberts, G. O. and Rosenthal, J. S. (2009). Examples of adaptive MCMC. \textit{J. Comput. Graph. Statist.} \textbf{18} 349--367.

\bibitem[Roberts and Sahu(1997)]{robsah97}
Roberts, G. O. and Sahu, S. K. (1997). Updating schemes, correlation structure, blocking and parameterization for the Gibbs sampler. \textit{J. Roy. Statist. Soc. Ser. B} \textbf{59} 291--317.

\bibitem[Sherlock and Thiery(2022)]{Sherlock2017}
Sherlock, C. and Thiery, A. H. (2022). A discrete bouncy particle sampler. \textit{Biometrika} \textbf{109} 335--349. https://doi.org/10.1093/biomet/asab013

\bibitem[Sherlock et al.(2015)]{Sherlock2015}
Sherlock, C., Thiery, A. H., Roberts, G. O. and Rosenthal, J. S. (2015). On the efficiency of pseudo-marginal random walk Metropolis algorithms. \textit{Ann. Statist.} \textbf{43}. https://doi.org/10.1214/14-AOS1278

\bibitem[Zanella and Roberts(2021)]{zanella2021multilevel}
Zanella, G. and Roberts, G. (2021). Multilevel linear models, Gibbs samplers and multigrid decompositions (with discussion). \textit{Bayesian Anal.} \textbf{16} 1309--1391.

\end{thebibliography}


\end{document}